\documentclass[12pt]{article}
\usepackage[a4paper]{geometry}
\usepackage[italian,english]{babel}

\usepackage{fancyhdr}
\usepackage{amsmath}

\usepackage{amsfonts}
\usepackage{amstext}
\usepackage{amssymb}
\usepackage{amsthm}
\usepackage{amscd}
\usepackage{soul}
\usepackage{comment}

\usepackage[pagebackref,draft=false]{hyperref}
\hypersetup{colorlinks,
linkcolor=myrefcolor,
citecolor=mycitecolor,
urlcolor=myurlcolor}

\usepackage[capitalize]{cleveref}
\usepackage{caption}

\usepackage{etaremune}

\usepackage{xcolor}
\definecolor{myurlcolor}{rgb}{0,0,0.4}
\definecolor{mycitecolor}{rgb}{0,0.5,0}
\definecolor{myrefcolor}{rgb}{0.5,0,0}
\usepackage{graphicx}
\usepackage{tikz}
\usepackage{tikz-cd}
\usepackage{mathrsfs}
\definecolor{ggray}{gray}{0.25}
 \definecolor{rred}{rgb}{0.5,0,0}

\usepackage{etoolbox}
\usepackage{makeidx}
\usepackage{sectsty}
\usepackage{dsfont}
\usepackage{enumitem} 
\usepackage[]{latexsym}
\usepackage{braket}
\usepackage{caption}
\usepackage[utf8]{inputenx}
\usepackage[T1]{fontenc}
\usepackage{lmodern}
\usepackage{textcomp}
\usepackage{microtype}
\usepackage{totcount}
\usepackage{blindtext}

\newtheorem{corollary}{Corollary}
\newtheorem{proposition}{Proposition}

\theoremstyle{definition}
\newtheorem{definition}{Definition}
\theoremstyle{remark}
\newtheorem{example}{Example}
\theoremstyle{remark}
\newtheorem{remark}{Remark}

\newcommand{\be}{\begin{equation}}
\newcommand{\ee}{\end{equation}}
\newcommand{\bea}{\begin{eqnarray}}
\newcommand{\eea}{\end{eqnarray}}
\newcommand{\vsp}{\vspace{0.4cm}}
\newcommand{\grit}[1]{{\bfseries {\itshape {#1}}}}

\newcommand{\rred}[1]{\color{rred}{#1}}

\newcommand{\ra}{\rightarrow}

\newcommand{\lra}{\longrightarrow}

\newcommand{\hh}{\mathcal{H}}
\newcommand{\bh}{\mathcal{B}(\mathcal{H})}

\newcommand{\Tr}{\textit{Tr}}

\newcommand{\stsp}{\mathscr{S}}

\newcommand{\appa}{\mathscr{A}}
\newcommand{\appas}{\mathscr{A}_{sa}}

\newcommand{\mappa}{\mathscr{M}}
\newcommand{\cappa}{\mathscr{C}}

\newcommand{\gapp}{\mathscr{G}}

\newcommand{\stav}{\mathscr{V}}

\newcommand{\pos}{\mathscr{P}}

\newcommand{\gr}{\mathrm{g}}

\newcommand{\Gg}{\mathrm{G}}

\title{Differential geometric aspects of parametric estimation theory for states on finite-dimensional $C^{\star}$-algebras}

\author{F. M. Ciaglia$^{1,3}$  \href{https://orcid.org/0000-0002-8987-1181}{\includegraphics[scale=0.7]{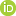}}, J. Jost$^{1,4}$\href{https://orcid.org/0000-0001-5258-6590}{\includegraphics[scale=0.7]{ORCID.png}}, L. Schwachh\"{o}fer$^{2,5}$\href{https://orcid.org/0000-0002-4268-6923}{\includegraphics[scale=0.7]{ORCID.png}}\\
\footnotesize{$^{1}$\textit{Max Planck Institute for Mathematics in the Sciences,  04103  Leipzig, Germany}} \\
\footnotesize{$^{2}$\textit{Faculty for Mathematics, TU Dortmund University,  44221  Dortmund, Germany}} \\
\footnotesize{$^{3}$\textit{ e-mail: \texttt{ciaglia[at]mis.mpg.de}, \texttt{florio.m.ciaglia[at]gmail.com}}} \\
\footnotesize{$^{4}$\textit{ e-mail: \texttt{jjost[at]mis.mpg.de}}}\\
\footnotesize{$^{5}$\textit{ e-mail: \texttt{lschwach[at]math.tu-dortund.de. }}}
}

\date{}

\begin{document}

\maketitle

\begin{abstract}
A geometrical formulation of estimation theory for  finite-dimensional $C^{\star}$-algebras is presented.
This formulation allows to deal with the classical and quantum case in a single, unifying mathematical framework.
The derivation of the Cramer-Rao and Helstrom bounds for parametric statistical models with discrete and finite outcome spaces is presented.
\end{abstract}

\tableofcontents

\thispagestyle{fancy}

 \section{Introduction}

The purpose of this work is to present the formulation of estimation theory in the framework of $C^{\star}$-algebras, with particular attention to  differential geometric aspects.
Although estimation theory is a well-developed subject both in the classical case of probability distributions \cite{AAVV1-1987,Amari-2016,A-N-2000,A-J-L-S-2017} and in the quantum case of density operators \cite{BN-G-J-2003,Helstrom-1967,Helstrom-1968,Helstrom-1969,Helstrom-1976,Holevo-2001,L-Y-L-W-2020,T-A-2014}, and even if quantum estimation theory builds upon classical estimation theory, there is no unifying picture for these subjects.
By unifying picture, we mean a mathematical framework in which estimation theory is placed in such a way that the classical and quantum cases appear as particular instances of the general theory.
We believe that such a formulation may be helpful in  obtaining a better understanding of the similarities and differences of classical and quantum estimation theory.
This idea may be considered as the driving force of this work.

Roughly speaking, the main problem estimation theory tries to address is to infer the value of some parameters characterizing the state of the ``physical system'' under investigation by the theoretical manipulation of the outcomes of experiments performed on such system.

In the classical case, the state of the system is described by a probability distribution on the space of outcomes of the experiment,  and the goal of estimation theory is to infer the value of some parameters characterizing the true probability distribution describing the system (e.g., the mean and/or variance for the case of Gaussian distributions) from the outcomes of the experiment.
As such, estimation theory is well-developed both in its asymptotic and non-asymptotic regimes.
Arguably, one little black spot of the theory is that the parameter spaces characterizing the probability distributions under study are usually taken to be homeomorphic to open subsets of some finite-dimensional Euclidean space.
Even if this assumption is justified in most of the models, it necessarily introduces some simplifications related with the ``nice'' structure of the parameter spaces as smooth manifolds.
As an example, the existence of global coordinates often lead to the definition of objects that are coordinate-dependent (see for instance  \cite[ch. 4]{A-N-2000} where it is clearly stated that the notion of unbiased estimator developed there is coordinate-dependent).
We believe it is healthy to formulate the theory in order to avoid these issues and better comprehend the coordinate-independent aspects of the theory.
This geometric attitude already proved itself useful in classical Newtonian, Lagrangian and Hamiltonian mechanics \cite{A-M-1978,B-C-T-2017,C-G-M-M-M-L-2006,M-F-LV-M-R-1990,M-M-1995,M-M-S-1984}, in  thermodynamics and statistical physics \cite{Barbaresco-2016,B-LM-N-2015,Marle-2016,Souriau-1978},  and in quantum mechanics \cite{C-M-P-1990,C-P-1990,E-I-M-M-2007,I-M-S-S-V-2016}.
Clearly, there already have been efforts to formulate classical estimation theory in this direction \cite{Cencov-1982,G-O-2006,Hendriks-1991,O-C-1995}, and here we try to encapsulate the spirit of these works in our formulation of estimation theory on $C^{\star}$-algebras.

On the other hand, in the quantum case, the state of the system is no-longer a probability distribution, it is a density operator on the Hilbert space associated with the system.
This adds, at the same time, complexity and richness to the problem of estimation.
A first layer of added complexity refers to the need of a statistical interpretation of a given quantum state.
Since the dawn of quantum mechanics, the issue of the physical interpretation of Schr\"{o}dinger's wavefunction was recognized to be a fundamental question.
The idea of interpreting the square modulus of the wavefunction as a probability distribution paved the way to the statistical interpretation of quantum states through  what is now called the \grit{Born rule} \cite{Landsman-2017}.
Essentially, the Born rule describes a ``procedure'' to associate a probability distribution on a suitable outcome space with a given quantum state.
Clearly, this depends on both the quantum state and the choice of the outcome space, and this means that there is more than one way to associated probability distributions with quantum states.
From the mathematical point of view, the choice of the statistical
interpretation is described by a positive operator-valued measure (POVM) on the Hilbert space of the system \cite{Holevo-2001,Kraus-1983}.
Accordingly, in order to set up the estimation problem for a given parametric model of quantum states, we need to operate a preliminary choice concerning the POVM ``inducing'' the statistical interpretation.
Of course, this choice influences the estimation problem we set up, and different choices in general lead to different solutions of the associated estimation problem.
All this obviously adds a layer of complexity to the estimation problem, but, simultaneously, it opens new possibilities to outperform classical limits of estimation because of the peculiar features of quantum states (e.g., entanglement).
Indeed, in the quantum case it is possible to give a precise mathematical meaning to the assertion ``measuring one copy N times is less informative than measuring N copies one single time'' \cite{Dowling-1998,G-L-M-2004,S-M-1995,T-A-2014,Y-MC-K-1986}.
This assertion relies on   the phenomenon of entanglement which is absent in the classical realm, and thus highlights an important difference between the classical and quantum estimation theory.

\vsp

As mentioned before, the goal of this article is to introduce a theoretical
framework that allows us to treat the classical and quantum case simultaneously.
Specifically, our choice is  to consider the theory of  $C^{\star}$-algebras  as the backbone of our construction because both probability distributions and quantum states may be realized as linear functionals on suitable $C^{\star}$-algebras.
In the case of probability distributions, this is basically the duality between probability measures and functions given by the Riesz theorem.
For quantum states, this comes directly from the axiomatic structure of the theory.
The main difference between the two cases is that the algebras involved are commutative in the former case and non-commutative in the latter.
In this general framework, probability distributions and quantum states represent different realizations of the notion of \grit{state} on a $C^{\star}$-algebra $\appa$.
The space of states $\stsp$ is a convex subset of the dual space of $\appa$, and the study of its differential geometry is a fascinating subject.
The rich algebraic structure of $C^{\star}$-algebras translate into a rich geometrical structure for their spaces of states that is perfectly suited for the formulation of parametric estimation theory.

The use of $C^{\star}$-algebras as a theoretical framework to study the   geometry of quantum states is not new \cite{B-Z-2006,C-C-I-M-V-2019,Ciaglia-2020,C-DC-I-L-M-2017,C-J-S-2020,G-H-P-2009,G-I-2001,G-I-2003,G-K-M-2005,Hasegawa-1995,Hasegawa-2003,H-P-1997,Jencova-2002,Jencova-2003,Petz-1993,Petz-1996,P-S-1996,P-T-1993,Uhlmann-1976,Uhlmann-2011}.
However, the focus was essentially always on the algebra of bounded linear operators on the Hilbert space of the quantum system, and not on a generic $C^{\star}$-algebra.
While this restriction may seem not particularly relevant for most practical purposes, it is certainly so from the theoretical point of view.
Indeed, some recent developments
\cite{C-DC-I-M-02-2020,C-DC-I-M-2020,C-I-M-2018,C-I-M-02-2019,C-I-M-03-2019,C-I-M-05-2019}
point out  the possibility of describing quantum systems whose  associated
$C^{\star}$-algebras are groupoid algebras, and thus are in principle more general than the algebra of bounded linear operators.
Consequently, a reformulation of the well-known results for an arbitrary $C^{\star}$-algebra appears to be useful.

On the other hand, in the classical case, the explicit use of $C^{\star}$-algebras to investigate the geometry of probability distributions is essentially absent.
To the best of our knowledge,  the only (very nice) exceptions are the works \cite{G-N-2020,G-R-2006,G-R-02-2006}.
However, the point of view  of these works is  different from ours because they consider probability distributions as particular elements of a $C^{\star}$-algebras, while we consider them as particular linear functionals.

\vsp

Another reason why we believe it would be useful to consider the framework of $C^{\star}$-algebras is that the space of states of a $C^{\star}$-algebra is an example of space of states of \grit{general probabilistic theories} \cite{Barrett-2007,C-DA-P}.
Therefore, the study  of the differential geometry of the space of states of $C^{\star}$-algebras, and in particular the study of parametric estimation theory in this context, represents a first step toward the study of these subjects in the more comprehensive frameork of general probabilistic theories.
This intermediate step may be useful because states on $C^{\star}$-algebras benefit from the rich algebraic structure of the algebras they act upon, while states in general probabilistic theories do not necessarily have such a rich algebraic background to rely on.
Consequently, a first study of the richer case may lead to results that can be later generalized to the less rich case once an appropriate and judicious process of extrapolation is pursued.

We confine ourselves to the case of finite-dimensional $C^{\star}$-algebras because, at this preliminary stage, we want to avoid the technical difficulties with which the infinite-dimensional case is filled.
Indeed, we are now interested in exposing the basic aspects of the theory in order to have a solid background on which future works can rely on.
In the infinite-dimensional case, the technical difficulties would often obscure the conceptual aspects and this unavoidably leads to be less comunicative.
Moreover, it is even not yet clear what are the geometrical players on the fields when infinite dimensions are considered because there is no general consensus on which are the most appropriate manifolds of states to consider in this case (see \cite{A-J-L-S-2017,A-J-L-S-2018,Bona-2004,C-I-J-M-2019,C-L-M-1983,C-M-P-1990,G-K-M-S-2018,G-P-1998,G-S-2000,Jencova-2006,Newton-2012,P-S-1995,Streater-2004} for some examples).

Incidentally, the restriction to the finite-dimensional case seems to affect more   the classical case, rather than  the quantum case.
Indeed, classical estimation theory essentially deals with parametric models of probability distributions on spaces which are neither discrete nor finite (think for instance to normal distributions), and these cases are naturally associated with infinite-dimensional $C^{\star}$-algebras.
The case of parametric models of probability distributions on discrete and finite spaces is usually less studied because it seldomly presents itself in applications.
In the context of quantum information theory, the situation is quite the opposite, and the vast majority of the models considered refer to quantum system with a finite-dimensional Hilbert space, and thus, with an associated finite-dimensional $C^{\star}$-algebra.
The infinite-dimensional case usually deals only with pure-state models for which the underlying manifold of states is rather friendly, being the Hilbert manifold of a complex projective space associated with a separable, complex Hilbert space.

\vsp

The content of this work builds on well-known and estabilished results in the context of both classical and quantum estimation theory.
However, the presentation of these results in the unifying framework of $C^{\star}$-algebras is essentially new, as are the proofs of some results.
We believe that this attitude may be particularly useful in future research dealing with the infinite-dimensional case, and dealing with the comparison of classical and quantum methods.
Accordingly, this work should be considered more as a first, preliminary step in a research program aimed at the understanding of the unification of classical and quantum estimation theory rather than an exposition of a finite theory,   and the focus of the work is more on the discussion of general structures rather than on the presentation of specific examples.

\vsp

The article  is structured as follows.
In section \ref{sec: differential geometry of states}, some differential geometric aspects of finite-dimensional $C^{\star}$-algebras and of their spaces of states are recalled.
In section \ref{sec: Parametric models of states on C-algebras}, the notion of parametric model of states on a  $C^{\star}$-algebra $\appa$ is introduced, and the notion of Symmetric Logarithmic Derivative used in quantum information theory is generalized to the $C^{\star}$-algebraic setting.
In section \ref{sec: statistical models}, the notion of parametric statistical  model associated with a given parametric model of states is introduced.
This notion represents the bridge between the models of states on a possibly noncommutative $C^{\star}$-algebra and the models of probability distributions used in classical estimation theory.
Also, the notion of multiple round model and its geometrical properties are briefly discussed.
In section \ref{sec: estimation theory}, the problem of estimation theory is formulated in the $C^{\star}$-algebraic framework, and the notion of manifold-valued estimator is recalled.
In section \ref{sec: cramer-rao bound for manifold-valued estimators}, a proof of the Cramer-Rao bound for manifold-valued estimators on finite outcome spaces is given following the work of Hendriks \cite{Hendriks-1991}.
Finally, in section \ref{subsec: helstrom bound for manifold-valued estimators}, the generalization of the Helstrom bound used in quantum information theory to the $C^{\star}$-algebraic framework is given.

\section{Differential geometric aspects of the space of states}\label{sec: differential geometry of states}

We start with a brief summary of $C^{\star}$-algebras \cite{Blackadar-2006,B-R-1987-1,Sakai-1971,Takesaki-2002}.
Let $A$ be a complex algebra with identity $\mathbb{I}$.
If there is an anti-linear map $\dagger\colon A\ra A$ such that $(\mathbf{a}^{\dagger})^{\dagger}=\mathbf{a}$ for all $\mathbf{a}\in A$, and such that $(\mathbf{ab})^{\dagger}=\mathbf{b}^{\dagger}\mathbf{a}^{\dagger}$ for all $\mathbf{a},\mathbf{b}\in A$, then $\dagger$ is called an \grit{involution}, and $(A,\dagger)$ an \grit{involutive algebra}.
If there is a norm $\|\cdot\|$ on $A$ turning it into a Banach space satisfying the additional relations $\|\mathbf{ab}\|\,\leq\,\|\mathbf{a}\|\,\|\mathbf{b}\|$ and $\|\mathbf{a}\mathbf{a}^{\dagger}\|\,\leq\,\|\mathbf{a}\|^{2}$ for all $\mathbf{a},\mathbf{b}\in A$, then  $(A,\dagger,\|\cdot\|)$ is called a $C^{\star}$\grit{-algebra}, and, for the sake of notational simplicity, it will be denoted simply by $\appa$.

An element $\mathbf{a}\in\appa$ is called \grit{self-adjoint} if $\mathbf{a}=\mathbf{a}^{\dagger}$.
The space of self-adjoint elements in $\appa$ is denoted as $\appa_{sa}$.
It is a real Banach space whose dual space is denoted as $\stav$,  and there is a direct sum decomposition
\begin{equation} \label{eq:decomp-A}
\appa = \appa_{sa} \oplus \imath\; \appa_{sa},
\end{equation}
where $\imath$ is the imaginary unit.
 

An element $\mathbf{b}\in\appa$ is called \grit{positive} if there exists $\mathbf{a}\in\appa$ such that $\mathbf{b}=\mathbf{a}^{\dagger}\mathbf{a}$.
Clearly, a positive element $\mathbf{b}$ is self-adjoint, and it can be proved that there is a unique self-adjoint element $\mathbf{s}$ such that $\mathbf{b}=\mathbf{s}^{2}$.

An element $\gr\in\appa$ is called \grit{invertible} if there is another element written as $\gr^{-1}$ such that $\gr\,\gr^{-1}=\gr^{-1}\,\gr=\mathbb{I}$.
The set of invertible elements in $\appa$ is denoted as $\gapp$, and it is a real Banach-Lie group the Banach-Lie algebra of which is $\appa$ endowed with the commutator \cite{Chu-2012,Upmeier-1985}.
An element $\mathbf{u}\in\gapp$ is called \grit{unitary} if $\mathbf{u}^{-1}=\mathbf{u}^{\dagger}$.
The set of unitary elements in $\appa$ is denoted as $\mathscr{U}$ and it is a real Banach-Lie subgroup of $\gapp$,  called the {\em unitary group of $\appa$}, whose  Banach-Lie algebra is the  subspace $\imath\; \appa_{sa}$ in the decomposition (\ref{eq:decomp-A})  
endowed with the commutator inherited from $\appa$ \cite{Chu-2012,Upmeier-1985}.

Let $\appa^{\star}$ be the complex Banach dual  of $\appa$.
An element $\xi\in\appa^{\star}$ is called a \grit{self-adjoint} linear functional if $\xi(\mathbf{a}^{\dagger})=\overline{\xi(\mathbf{a})}$.
The set of self-adjoint linear functionals is precisely the real Banach dual $\stav$ of $\appas$.
A self-adjoint linear functional $\omega$ is called \grit{positive} if $\omega(\mathbf{a})\geq 0$ for every positive element $\mathbf{a}\in\appa$.
A positive element $\omega$ is called \grit{faithful} if $\omega(\mathbf{a})=0$ implies $\mathbf{a}=\mathbf{0}$ for all positive elements in $\appa$.
The set of positive elements is denoted as $\pos$.
A positive linear functional $\rho$ is called a \grit{state} if $\rho(\mathbb{I})=1$.
The set of states is denoted as $\stsp$.
 
In the following, we will focus only on finite-dimensional $C^{\star}$-algebras.
Given a self-adjoint element $\mathbf{a}\in\appas$, we write $f_{\mathbf{a}}$ for the linear function on $\stav$ given by
\be\label{eqn: linear functions}
f_{\mathbf{a}}(\xi)\,=\,\xi(\mathbf{a})\,,
\ee
as well as for its restrictions to the various submanifolds of $\stav$ we will introduce below (with an evident abuse of notation).

There is a  group action of $\gapp$ on $\stsp$ given by \cite{C-I-J-M-2019,C-J-S-2020}
\be\label{eqn: action of G on states}
\rho_{\gr}\,\colon\;\;\rho_{\gr}(\mathbf{c})\,=\,\frac{\rho(\gr^{\dagger}\,\mathbf{c}\,\gr)}{\rho(\gr^{\dagger}\,\gr)}\,\equiv\,\Phi(\gr,\rho)\;\;\;\;\forall \;\mathbf{c}\in\appa ,
\ee
and the space of states $\stsp$ decomposes into the disjoint union of orbits of the $\gapp$-action, and evidently, each such orbit is a homogeneous space.
 

Recalling that $\appa$ endowed with the commutator is the Lie algebra of $\gapp$, the fundamental vector fields of $\Phi$ are labelled by elements of $\appa$.
 Recalling \eqref{eq:decomp-A}, we  write an element in $\appa$ as
$\mathbf{a} +\imath \mathbf{b}$ where $\mathbf{a},\mathbf{b}\in\appas$, and
$\imath$ is the imaginary unit.
Accordingly, we write $\Gamma_{\mathbf{ab}}$ for the fundamental vector field associated with $\frac{1}{2}\left(\mathbf{a} +\imath \mathbf{b}\right)$.
A direct computation shows that the tangent vector $\Gamma_{\mathbf{ab}}(\rho)$, identified with a self-adjoint linear functional in $\stav$ because the orbit $\mathcal{O}$ is an immersed submanifold of $\stav$, is given by
\be
\left(\Gamma_{\mathbf{ab}}\,f_{\mathbf{c}}\right)(\rho)\,=\,\left(\Gamma_{\mathbf{ab}}(\rho)\right)(\mathbf{c}) \,=\,\rho(\{\mathbf{a},\mathbf{c}\}) - \rho(\mathbf{a})\rho(\mathbf{c}) + \rho([[\mathbf{b},\mathbf{c}]])\;\;\;\;\forall \;\mathbf{c}\in\appa,
\ee
where   $\{,\}$ and $[[,]]$ denote, respectively, the Jordan product and the Lie product in $\appa$ given by 
\be\label{eqn: Lie and Jordan products}
\begin{split}
\{\mathbf{c},\mathbf{d}\}&\,:=\,\frac{1}{2}\,\left(\mathbf{cd} + \mathbf{dc}\right) \\
[[\mathbf{c},\mathbf{d}]]&\,:=\,\frac{1}{2\imath}\,\left(\mathbf{cd} - \mathbf{dc}\right).
\end{split}
\ee
Note that $\{,\}$ and $[[,]]$ preserve $\appas$, and actually turn it into a Jordan-Lie algebra \cite{F-F-I-M-2013,Landsman-1998}.

We set
\be\label{eqn: hamiltonian and gradient vector fields}
\begin{split}
\mathbb{Y}_{\mathbf{a}}&\,:=\,\Gamma_{\mathbf{a}\mathbf{0}} \\
\mathbb{X}_{\mathbf{b}}&\,:=\,\Gamma_{\mathbf{0}\mathbf{b}},
\end{split}
\ee
and we call $\mathbb{Y}_{\mathbf{a}}$ a gradient vector field (the origin of the name will be explained below), and $\mathbb{X}_{\mathbf{b}}$ a Hamiltonian vector field.
It is not hard to show that the Hamiltonian vector fields give an anti-representation of the Lie algebra of the group $\mathscr{U}\subset\gapp$ of unitary element of $\gapp$ \cite{C-I-J-M-2019,C-J-S-2020}.
This Lie algebra anti-representation integrates to a left action of $\mathscr{U}$ on $\mathcal{O}$ given by the restriction of $\Phi$ to $\mathscr{U}$.

If we fix a basis $\{\mathbf{e}^{j}\}_{j=1,...,N}$ of self-adjoint elements in $\appa$ (where $\mathrm{dim}(\appa)=N$), we may introduce the structure constants $d^{jk}_{l}$ and $c^{jk}_{l}$ of the Jordan and Lie products in equation  \eqref{eqn: Lie and Jordan products} by setting
\be
\begin{split}
\{\mathbf{e}^{j},\mathbf{e}^{k}\}&\,=\,d^{jk}_{l}\,\mathbf{e}^{l} \\
[[\mathbf{e}^{j},\mathbf{e}^{k}]]&\,:=\,c^{jk}_{l}\,\mathbf{e}^{l}.
\end{split}
\ee
Then, the gradient and Hamiltonian vector fields are easily seen to be given by
\be
\begin{split}
\mathbb{Y}_{\mathbf{a}}&\,=\,\left(d^{jk}_{l}a_{k}x^{l} - f_{\mathbf{a}}\,x^{j} \right)\,\frac{\partial}{\partial x^{j}} \\
\mathbb{X}_{\mathbf{b}}&\,:=\,c^{jk}_{l}b_{k}x^{l}\,\frac{\partial}{\partial x^{j}},
\end{split}
\ee
where $\{x^{j}\}_{j=1,...,N}$ is the Cartesian coordinate system on $\stav$ associated with the dual basis $\{\mathbf{e}_{j}\}_{j=1,...,N}$ of $\{\mathbf{e}^{j}\}_{j=1,...,N}$.

\begin{example}[The probability simplex]\label{ex: the probability simplex}
If we endow $C(\mathcal{X}_{n})$   with the involution given by complex
conjugation, and with the supremum norm, it is not hard to prove that it is a $C^{\star}$-algebra.
We denote this $C^{\star}$-algebra as $\cappa_{n}$.
Let $\mathbf{e}^{j}\in\cappa_{n}$ be the `delta function' at
$\mathbf{x}_{j}\in\mathcal{X}_{n}$ (i.e.,
$\mathbf{e}^{j}(\mathbf{x}_{k})=\delta_{k}^{j}$), then
$\{\mathbf{e}^{j}\}_{j=1,..,n}$ is clearly a basis for $\cappa_{n}$ (seen as a
vector space) made up of positive, self-adjoint elements, and we have
\be
\sum_{j=1}^{n}\,\mathbf{e}^{j}\,=\,\mathbf{1}_{n},
\ee
where $\mathbf{1}_{n}$ is the identity element in $\cappa_{n}$ (i.e., the identity function on $\mathcal{X}_{n}$). 
Consequently, we can build the dual basis $\{\mathbf{e}_{j}\}_{j=1,..,n}$, and a state $\rho$ on $\cappa_{n}$ is easily seen to be written as 
\be\label{eqn: dual basis for cappan}
\rho\,=\,p^{j}\,\mathbf{e}_{j},
\ee
where the real numbers $p^{j}=\rho(\mathbf{e}^{j})$ are non-negative and are subject to the constraint
\be
\sum_{j=1}^{n}\,p^{j}\,=\,1.
\ee
From this, we conclude that the space of states $\stsp$ of $\cappa_{n}$ may be identified with the $n$-simplex $\Delta_{n}$.
In the following, whenever we deal with $\cappa_{n}$ we will identify a state $\rho$ on $\cappa_{n}$ with a probability distribution in $\Delta_{n}$ and write $\mathbf{p}$ instead of $\rho$.

Let $I_{k}\subseteq\mathcal{X}_{n}$ be a subset with $k\leq n$ elements, and let $\rho$ be a state on $\cappa_{n}$ such that $p^{j}\neq 0$ if and only if $\mathbf{x}_{j}\in I_{k}$.
Then, it is not hard to check that the orbit $\mathcal{O}$ of the group $\gapp$ of invertible elements in $\cappa_{n}$ (see equation \eqref{eqn: action of G on states}) through $\rho$  coincides with the set of all those states $\varrho=q^{j}\mathbf{e}_{j}$ such that $q^{j}\neq 0$ if and only if $\mathbf{x}_{j}\in I_{k}$.
In particular, the open interior $\Delta_{n}^{+}$ of the $n$-simplex may be identified with the orbit of $\gapp$ through the state $\mathbf{p} $ with $p^{j} =\frac{1}{n}$ for all $j=1,...,n$. 

Since $\cappa_{n}$ is Abelian, it is not hard to see that the action of the unitary group $\mathscr{U}\subset\gapp$ is trivial, and thus the Hamiltonian vector fields vanish identically.
On the other hand, a direct computation shows that the structure constants $d^{jk}_{l}$ of the Jordan product  with respect to the basis $\{\mathbf{e}^{j}\}_{j=1,..,n}$ vanish unless $j=k=l$, in which case they are $1$.
\end{example}

\begin{example}[The space of density matrices]\label{ex: the space of density operators}

Consider the complex algebra $\mappa_{n}:=\mathcal{M}_{n}(\mathbb{C})$ of complex-valued, $(n\times n)$ matrices.
There is an involution $\dagger$ on $\mappa_{n}$ given by the composition of transposition with component-wise complex conjugation.
By exploiting the trace operation, it is possible to define a norm on $\mappa_{n}$ given by $\|\mathbf{a}\|^{2}\:=\Tr(\mathbf{a}^{\dagger}\mathbf{a})$, and we obtain a $C^{\star}$-algebra which will be denoted by $\mappa_{n}$.
Moreover, it is easily seen that $\mappa_{n}$ is isomorphic to the algebra $\mathcal{B}(\hh)$ of bounded linear operators on an $n$-dimensional complex Hilbert space $\hh$.
This isomorphism depends on the choice of an orthonormal basis in $\hh$, but, in the context of quantum information theory, this is in general not very limiting because a preferred choice of basis, called \grit{computational basis} \cite{N-C-2011}, is often tied to the physics of the problem under investigation.

Since $\mappa_{n}$ is finite-dimensional, it is isomorphic with its dual space, and an isomorphism is provided by the trace operation.
Specifically, a linear functional $\xi$ on $\mappa_{n}$ is identified with an element $\widehat{\xi}\in\mappa_{n}$ by means of
\be
\xi(\mathbf{a})\,:=\,\Tr(\widehat{\xi}\,\mathbf{a}).
\ee
Then, it follows that a state $\rho$  on $\mappa_{n}$ may be identified with a positive semi-definite matrix $\widehat{\rho}\in\mappa_{n}$ with unit trace.
Any such matrix is usually referred to as a \grit{density matrix}.

It is not hard to prove that the orbits of $\gapp$ are classified by the rank of the associated density matrices \cite{C-C-I-M-V-2019,C-DC-I-L-M-2017,C-I-J-M-2019,G-K-M-2005,G-K-M-2006}.
Specifically, every orbit $\mathcal{O}$ is made up of states the associated density matrices of which have fixed rank.
In particular, we have the orbit of states whose density matrices have unit rank which is the orbit of \grit{pure states} (the extremal points of the convex space of states) which is diffeomorphic to the complex projective space $\mathbb{CP}^{n}$, and the orbit of states whose density matrices have full rank (invertible) which is the orbit of \grit{faithful states}.
Note that the latter is an open subset of the affine space of self-adjoint linear functionals giving 1 when evaluated on the identity $\mathbb{I}_{n}$ of $\mappa_{n}$.

If we introduce a basis $\{\sigma^{j}\}_{j=0,...,n^{2}-1}$ on $\mappa_{n}$ in such a way that $\sigma^{0}$ coincides with the identity element $\mathbb{I}_{n}\in\mappa_{n}$, and that $\sigma^{j}$ is self-adjoint and satisfies $\Tr(\sigma^{j})=0$ for all $j\neq 0$, we can build its dual basis $\{\sigma_{j}\}_{j=0,...,n^{2}-1}$, and it follows that a state $\rho$ may be written as
\be
\rho\,=\,\frac{1}{n} \left(\sigma_{0} + x^{j}\,\sigma_{j}\right),
\ee
where $j=1,...,n^{2}-1$, and $x^{j}\in\mathbb{R}$.
Clearly, the fact that $\rho$ must be a state  imposes some constraints on the values of $x^{j}$ depending on the fact that $\rho(\mathbf{a}^{\dagger}\mathbf{a})$ must be non-negative.
There is no general closed formula to express these constraints for arbitrary $n> 2$.

For the case $n=2$ (also known as the \grit{qubit}), it is customary to select $\sigma^{1},\sigma^{2},\sigma^{3}$ to be the so-called Pauli matrices
\be
\sigma^{1}=\left(\begin{matrix} 0 & 1 \\ 1 & 0 \end{matrix}\right) \qquad \sigma^{2}=\left(\begin{matrix} 0 & -\imath \\ \imath & 0 \end{matrix}\right) \qquad \sigma^{3}=\left(\begin{matrix} 1 & 0 \\ 0 & -1 \end{matrix}\right),
\ee
where $\imath$ is the imaginary unit.
Then, $\rho$ is a state  if and only if
\be
\delta_{jk}x^{j}x^{k}\,\leq\,1\,.
\ee
This identifies a ball in the three-dimensional space spanned by the Pauli matrices which is known as the \grit{Bloch ball}.
In this case, there are only two orbits of the group $\gapp$ of invertible elements in $\mappa_{n}$, namely, the density matrices lying on the surface sphere (the \grit{pure states}), and the density matrices in the interior of the ball (the \grit{faithful states}).
\end{example}

According to \cite{C-J-S-2020}, the gradient vector fields  provide an overcomplete basis of the tangent space at each point in every orbit  $\mathcal{O}$.
Furthermore, on  every  $\mathcal{O}$ we may define a Riemannian metric tensor $\Gg$   given by
\be\label{eqn: Jordan metric tensor}
\Gg_{\rho}(\mathbb{Y}_{\mathbf{a}}(\rho),\mathbb{Y}_{\mathbf{b}}(\rho))\,=\,\rho(\{\mathbf{a},\mathbf{b}\}) - \rho(\mathbf{a})\rho(\mathbf{b})\,,
\ee  
and $\mathbb{Y}_{\mathbf{a}}$ is the gradient vector field associated with the smooth function $f_{\mathbf{a}}$ (see equation \eqref{eqn: linear functions}) by means of $\Gg$.
This metric tensor is invariant with respect to the action of the unitary group in the sense that
\be\label{eqn: unitary invariance of G}
\Phi_{\mathbf{U}}^{*}\Gg\,=\,\Gg \qquad \forall\,\mathbf{U}\in\mathscr{U},
\ee
where $\Phi_{\mathbf{U}}$ is the diffeomorphism given by
\be
\Phi_{\mathbf{U}}(\rho)\,:=\,\Phi(\mathbf{U},\rho)\,.
\ee
However, $\Gg$ is not invariant under the action of all of $\gapp$.

The metric tensor $\Gg$ turns out to be the $C^{\star}$-algebraic version of some well-known and relevant metric tensors when explicit cases are considered \cite{C-J-S-2020}. 
For instance, if $\appa=\cappa_{n}$ and $\mathcal{O}=\Delta_{n}^{+}$, then $\Gg$ coincides with the Fisher-Rao metric tensor.
If $\appa=\bh$ and $\mathcal{O}\cong\mathbb{CP}(\hh)$ is the orbit of pure states), then $\Gg$ coincides with the Fubini-Study metric tensor.
If $\appa=\bh$ and $\mathcal{O}$ is the orbit of faithful states, then $\Gg$ coincides with the Bures-Helstrom metric tensor.

According to \cite{C-J-S-2020}, the geodesic of $\Gg$ starting at $\rho\in\mathcal{O}$ with initial tangent vector $\mathbf{v}\in T_{\rho}\mathcal{O}$ reads
\be
\nu_{\rho}^{\mathbf{v}}(t)\,=\,\cos^{2}(|\mathbf{v}|t)\,\rho + \frac{\sin^{2}(|\mathbf{v}|t)}{|\mathbf{v}|^{2}}\,\rho_{\mathbf{v}} + \frac{\sin(2|\mathbf{v}|t)}{2|\mathbf{v}|}\,\rho_{\{\mathbf{v}\}},
\ee
where
\be
\begin{split}
\mathbf{v}\,=\,\mathbb{Y}_{\mathbf{a}}(\rho)\quad\mbox{ for some } &\,\mathbf{a}\in\appas\;|\,\,\rho(\mathbf{a})\,=\,0,\\
|\mathbf{v}|^{2}\,=\,\Gg_{\rho}(\mathbf{v},\mathbf{v})&\,=\,\rho\left(\mathbf{a}^{2}\right),\\
\rho_{\mathbf{v}}(\mathbf{b})\,:=\,\rho\left(\mathbf{a}\,\mathbf{b}\,\mathbf{a}\right)\quad&\forall\,\,\mathbf{a}\in\appas, \\
\rho_{\{\mathbf{v}\}}(\mathbf{b})\,:=\,\rho\left(\{\mathbf{a},\,\mathbf{b}\}\right)\quad&\forall\,\,\mathbf{a}\in\appas.
\end{split}
\ee
The geodesic $\nu_{\rho}^{\mathbf{v}}(t)$ remains inside the space of states $\stsp$ for all $t\in\mathbb{R}$, but it also exits and enters the orbit $\mathcal{O}$ containing the initial state $\rho$ at multiple times \cite{C-J-S-2020}.

\section{Parametric models of states on $C^{\star}$-algebras}\label{sec: Parametric models of states on C-algebras}

Motivated by the classical theory of parametric estimation, we will now introduce the notion of a parametric model of states on a finite-dimensional $C^{\star}$-algebra, and then reformulate the theory of parametric estimation in this theoretical framework.
This will allow for the simultaneous handling of the classical and the quantum case.

\begin{definition}\label{def: parametric model of states}
A \grit{parametric  model} of states on a (finite-dimensional) $C^{\star}$-algebra $\appa$ is a triplet  $(M,\mathit{j},\mathcal{O})$ where $M$ is a smooth manifold, $\mathcal{O}\subset \stsp$ is a $\gapp$-orbit in $\stsp$  (see section \ref{sec: differential geometry of states}), and  $\mathit{j}\colon M\ra \mathcal{O}$ is a smooth map.
If $\mathit{j}$ is injective, we say that the model is \grit{identifiable}.
\end{definition}


Some comments are in order.
First of all, we fix the codomain of $\mathit{j}$ to be an orbit of states $\mathcal{O}$ because, as will be clear below, we want to exploit the  differential geometric aspects of $\mathcal{O}$ itself.
In practice, a vast part of the models considered in the literature falls in this category.
For instance, in quantum information geometry, it is customary to deal with parametric models consisting only of pure states, or only of invertible density operators.
In principle, it would also be possible to consider a more general case in which $\mathit{j}$ is a smooth map of $M$ into the Banach space  $\stav$ of self-adjoint linear functionals in such a way that   $\mathit{j}(M)\subset\stsp$, and $M$ intersects different orbits of states.
This line of thought would require a different way to handle geometrical properties of the space of states in relation with the parameter manifold, based, for example, on the methodology introduced in \cite{A-J-L-S-2017,A-J-L-S-2018} for the classical case.
This line of reasoning may be useful in the transition to the infinite-dimensional case where the smooth structure of the orbits $\mathcal{O}$ is in general not guaranteed, and we plan to address this and related questions in the future.

Concerning the identifiability of a model, it may seem at first glance reasonable to consider only identifiable models, but we will show that there are well-known and ``simple'' parametric models of quantum states (e.g., qubit models) for which either this assumption is not satisfied, or it leads to difficulties with the statistical interpretation of the model.

\vsp

Now, we turn our attention to the geometrical objects that $M$ inherits by means of the smooth map $\mathit{j}$.
Indeed, once we have the smooth map $\mathit{j}$, a symmetric, covariant $(0,2)$ tensor is naturally obtained on $M$ by considering the Riemannian metric $\Gg$ on $\mathcal{O}$ introduced before  and   taking its pullback 
\be\label{eqn: pullback metric}
\Gg^{M}\,:=\mathit{j}^{*}\Gg
\ee
to $M$ with respect to  $\mathit{j}$.
This gives a tensor on $M$ which ``feels'' the possible non-commutativity of $\appa$ and gives the ``correct'' tensor in the classical case.

Indeed, if $\appa$ is Abelian, then $\mathcal{O}$ is diffeomorphic to the open interior of a suitable simplex, $\Gg$ is the Fisher-Rao metric tensor \cite{C-J-S-2020}, and $\Gg^{M}$  is the  pullback of the Fisher-Rao metric tensor to the manifold $M$ seen as a model of probability distributions  \cite{A-N-2000}.

On the other hand, if $\appa$ is the algebra $\bh$ of bounded linear operators on a finite-dimensional, complex Hilbert space $\hh$ and $\mathcal{O}$ is the manifold of pure states, then $\mathcal{O}$ is diffeomorphic to the complex projective space $\mathbb{CP}(\hh)$ associated with $\hh$, $\Gg$ is the Fubini-Study metric \cite{C-J-S-2020}  on $\mathcal{O}=\mathbb{CP}(\hh)$, and $\Gg^{M}$ is the quantum counterpart  of the Fisher-Rao metric tensor   on the manifold $M$ seen as a model of pure quantum states \cite{F-K-M-M-S-V-2010}.
Also, if $\mathcal{O}$ is the manifold of faithful states, then $\Gg$ is the
Bures-Helstrom metric tensor \cite{C-J-S-2020}, and $\Gg^{M}$ may be read as a
quantum counterpart of the Fisher-Rao metric tensor   on the manifold $M$ seen as a  model of faithful quantum states \cite{Paris-2009}.

\vsp

We will now introduce the $C^{\star}$-algebraic version of the Symmetric Logarithmic Derivative (SLD)   introduced in quantum estimation theory by Helstrom in \cite{Helstrom-1967}.
 For this purpose, note that  every tangent vector at $\rho\in\mathcal{O}$ may be expressed in terms of gradient vector fields, that is, given $\rho\in\mathcal{O}$,  for every tangent vector $V_{\rho}\in T_{\rho}\mathcal{O}$ there exists a self-adjoint element $\mathbf{a}\in\appas$ depending on $V_{\rho}$ such that
\be
V_{\rho}\,=\,\mathbb{Y}_{\mathbf{a}}(\rho)\,.
\ee
Consequently, if we consider a tangent vector $v_{m}\in T_{m}M$, it makes sense to ask for the gradient vector field $\mathbb{Y}_{\mathbf{a}}$ on $\mathcal{O}$ such that
\be\label{eqn: SLD}
T_{m}\mathit{j}(v_{m})\,=\,\mathbb{Y}_{\mathbf{a}}(\rho_{m})\,,
\ee
where $\rho_{m}:=\rho(j(m))$.
The gradient vector field $\mathbb{Y}_{\mathbf{a}}$ in general depends on both  the point $m\in M$ and the tangent vector $v_{m}$.
The tangent vector $\mathbb{Y}_{\mathbf{a}}(\rho_{m})$ satisfying equation \eqref{eqn: SLD} is called the SLD of $v_{m}$ at $\rho_{m}$.

To appreciate the link with the standard definition of the SLD, let us consider a parametric model $(M,\mathit{j},\mathcal{O})$ where  $\appa=\bh$, $\mathcal{O}$ is the manifold of faithful states (invertible density operators), $M$ is an open submanifold of $\mathbb{R}$, and $\mathit{j}$ is a suitable smooth map.
Setting $v_{m}=\partial_{t}(m)$ where $\partial_{t}$ is the restriction to $M$ of the vector field generating the group structure of $\mathbb{R}$,  a direct computation shows that the solution of equation \eqref{eqn: SLD} coincides with the Symmetric Logarithmic Derivative (SLD) of \cite{Helstrom-1967}.
Indeed, $\partial_{t}$ is the infinitesimal generator of $m_{t}=m + t$, and considering an arbitrary function $f_{\mathbf{b}}$ on $\mathcal{O}$, we have
\be
\langle \mathrm{d}f_{\mathbf{b}}(\rho_{m}),\,T_{m}\mathit{j}(v_{m})\rangle\,=\,\frac{\mathrm{d}}{\mathrm{d}t}\,\left(\mathrm{Tr}\left(\hat{\rho}_{m_{t}}\,\mathbf{b}\right)\right)_{t=0}\,=\,\mathrm{Tr}\left(\frac{\mathrm{d}}{\mathrm{d}t}\,\left(\hat{\rho}_{m_{t}}\right)_{t=0}\,\mathbf{b}\right)\;\;\;\forall\,\mathbf{b}\in\appas
\ee
so that equation \eqref{eqn: SLD} may be alternatively written as
\be
\frac{\mathrm{d}}{\mathrm{d}t}\,\left(\hat{\rho}_{m_{t}}\right)_{t=0}\,=\,\{\hat{\rho}_{m},\,\mathbf{a}_{m}\}\,=\,\frac{1}{2}\left(\hat{\rho}_{m}\,\mathbf{a}_{m} + \mathbf{a}_{m}\,\hat{\rho}_{m}\right),
\ee
where 
\be
\mathbf{a}_{m}\,=\,\mathbf{a} - \hat{\rho}_{m}(\mathbf{a})\,\mathbb{I},
\ee
which is precisely the definition of the SLD (see also equation 3 in \cite{Paris-2009} and equations 3.4 and 3.14 in \cite{Tsang-2019} for the multiparametric case).
This justifies the interpretation of equation \eqref{eqn: SLD} as the $C^{\star}$-algebraic generalization of the SLD embracing also the multiparametric quantum and classical cases.

\begin{example}[A pure state qubit model]\label{ex: the circle in the space of pure states of a qubit}

Consider the algebra $\mappa_{2}$ of the qubit (see example \ref{ex: the space of density operators}).
Take the one-parameter group of unitary elements generated by the element $\imath\sigma_{3}$ according to 
\be\label{eqn: gamma qubit}
\mathbf{u}_{\gamma}\,=\,\mathrm{e}^{\frac{\imath}{2}\gamma\sigma^{3}},
\ee
where $\gamma\in\mathbb{R}$.
Then, consider the orbit $\mathcal{O}\cong\mathbb{CP}^{2}$ of pure states on $\mappa_{2}$, set $M=\mathbb{R}$, and consider the   map $\mathit{j}_{\mathbb{R}}\colon M\ra\mathcal{O}$ given by
\be\label{eqn: gamma qubit 2}
\rho_{\gamma}\,\equiv\,\mathit{j}_{\mathbb{R}}(\gamma)\,:=\,\Phi(\mathbf{u}_{\gamma},\rho),
\ee
where $\Phi$ is the action of $\gapp\supset\mathscr{U}$ given in equation \eqref{eqn: action of G on states}, and 
\be\label{eqn: pure state qubit model R 1}
\rho\,=\,\frac{1}{2}\left(\sigma_{0} + \sigma_{1}\right).
\ee
A direct computation shows that 
\be\label{eqn: pure state qubit model R 2}
\rho_{\gamma}\,=\,\frac{1}{2}\left(\sigma_{0} + \cos(\gamma)\,\sigma_{1} - \sin(\gamma)\,\sigma_{2} \right)
\ee
and that $\mathit{j}_{\mathbb{R}}$ is smooth.
Clearly, $\mathit{j}_{\mathbb{R}}$ is not injective, and thus the parametric model  $(\mathbb{R},\mathit{j}_{\mathbb{R}},\mathbb{CP}^{2})$ is not identifiable.
However, the parametrization given in equation \eqref{eqn: pure state qubit model R 2} is useful in quantum estimation theory when an experimental realization of the parametric model is constructed in terms of a spin interacting with a magnetic field.
In this case,  $\gamma=tB$ where $t$ is the time parameter of the dynamical evolution and $B$ is the strenght of the magnetic field.
Then, the fact that the model is not identifiable depends on the dynamical evolution being periodic.

Now, let us  consider the vector field $V$ on $M=\mathbb{R}$ generating translations.
This vector field is complete, and provide a basis of the tangent space $T_{\gamma}M$ at each $\gamma\in M$.
Moreover, $V$ is the infinitesimal generator of the action of the Abelian Lie group $G=\mathbb{R}$ on $M=\mathbb{R}$ given by
\be
\psi(\zeta,\gamma)\,:=\,\gamma + \zeta \qquad \forall\,\zeta\in G,\,\gamma\in M.
\ee
The group $G$ acts also on $\mathbb{CP}^{2}$ by means of
\be
\Psi(\zeta,\rho)\,:=\,\Phi(\mathbf{U}_{\zeta},\rho),
\ee
where $\Phi$ is the action given in equation \eqref{eqn: action of G on states}.
The fact that $\Psi$ is a group action follows from the fact that the map $\gamma\mapsto \mathbf{U}_{\gamma}$ is   a group homomorphism, that is, it satisfies
\be
\mathbf{U}_{\zeta_{1}}\,\mathbf{U}_{\zeta_{2}}\,=\,\mathbf{U}_{\zeta_{1} + \zeta_{2}}\,\qquad \forall\,\zeta_{1},\zeta_{2}\in G.
\ee
The actions $\psi$ and $\Psi$ have a particular relation to one another, indeed, a direct computation shows that they are equivariant with respect to $\mathit{j}_{\mathbb{R}}$, which means that
\be\label{eqn: relatedness of actions on pure state qubit model}
\mathit{j}_{\mathbb{R}}\left(\psi(\zeta,\gamma)\right)\,=\,\Psi\left(\zeta,\mathit{j}_{\mathbb{R}}(\gamma)\right)\,.
\ee
This property is quite strong because it implies that the fundamental vector fields of the action of $G$ on $M=\mathbb{R}$ are $\mathit{j}_{\mathbb{R}}$-related with the fundamental vector fields of the action of $G$ on $\mathbb{CP}^{2}$, which means that \cite{A-M-R-1988}
\be
T_{\gamma}\mathit{j}_{\mathbb{R}}(V_{\gamma})\,=\,W_{\rho_{\gamma}}^{\zeta},
\ee
where $V$ is the fundamental vector field of $\psi(\zeta,\gamma)$ (i.e, the vector field generating the translation  considered above), while $W$ is the fundamental vector field of $\Psi(\zeta,\rho)$ (recall that, in this case, the exponential map from the Lie algebra of $G$ to $G$ itself is the identity).
Since
\be
\Psi(\zeta,\rho)\,=\,\Phi(\mathbf{U}_{\zeta},\rho),
\ee
the fundamental vector field $W$ is easily seen to be the Hamiltonian vector field associated with $\sigma^{3}$ (see equation \eqref{eqn: hamiltonian and gradient vector fields}).
This means that
\be
\langle \mathrm{d}f_{\mathbf{b}},\,T_{\gamma}\mathit{j}_{\mathbb{R}}(V_{\gamma})\rangle\,=\, \langle \mathrm{d}f_{\mathbf{b}},\,W_{\rho_{\gamma}}\rangle\,=\, \rho_{\gamma}\left([[\sigma^{3},\,\mathbf{b}]]\right)\,.
\ee
Consequently, regarding the SLD, equation \eqref{eqn: SLD} leads us to look for the self-adjoint element $\mathbf{a}$ satisfying
\be\label{eqn: SLD qubit}
\rho_{\gamma}\left(\{\mathbf{a},\,\mathbf{b}\}\right) - \rho_{\gamma}(\mathbf{a})\,\rho_{\gamma}(\mathbf{b})\,=\, \rho_{\gamma}\left([[\sigma^{3},\,\mathbf{b}]]\right)\,
\ee
for all self-adjoint elements $\mathbf{b}\in\mappa_{2}$.
Passing from $\rho_{\gamma}$ to its density matrix $\hat{\rho}_{\gamma}$, we see that equation \eqref{eqn: SLD qubit} is equivalent to
\be\label{eqn: SLD qubit 2}
\{\hat{\rho}_{\gamma},\,\mathbf{a}\} - \Tr(\hat{\rho}_{\gamma}\mathbf{a})\,\hat{\rho} \,=\,[[\hat{\rho}_{\gamma},\,\sigma^{3}]].
\ee
We write
\be
\mathbf{a}\,=\,a_{0}\sigma^{0}  + a_{1}\sigma^{1} + a_{2}\sigma^{2} + a_{3}\sigma^{3},
\ee
where $a_{j}\in\mathbb{R}$ for all $j=0,1,2,3$.
A direct computation exploiting the properties of the Pauli matrices shows that  $a_{0}$ is arbitrary (as it should be because of the very definition of gradient vector field), $a_{3}=0$, while $a_{1}$ and $a_{2}$ must satisfy
\be
a_{1}\sin(\gamma) + a_{2}\cos(\gamma)\,=\,-1 .
\ee
Clearly, this means that $\mathbf{a}$ and thus the SLD are not uniquely defined.

Concerning the covariant tensor $\Gg^{\mathbb{R}}$, we have
\be
\Gg^{\mathbb{R}}\,=\,\mathit{j}_{\mathbb{R}}^{*}\Gg
\ee
by definition.
Since $\mathit{j}_{\mathbb{R}}$ is an immersion and $\Gg$ is a Riemannian metric, then $\Gg^{\mathbb{R}}$ is a Riemannian metric (i.e., it is positive and invertible).
Moreover, setting
\be
\begin{split}
\psi_{\zeta}(\gamma)&\,:=\,\psi(\zeta,\gamma) \\
\Psi_{\zeta}(\rho)&\,:=\,\Psi(\zeta,\rho)\,=\,\Phi(\mathbf{U}_{\zeta},\rho),
\end{split}
\ee
we immediately obtain
\be
\psi_{\zeta}^{*}\Gg^{\mathbb{R}}\,=\,\psi_{\zeta}^{*}\mathit{j}_{\mathbb{R}}^{*}\Gg\,=\,\left(\mathit{j}_{\mathbb{R}}\circ\psi_{\zeta} \right)^{*}\Gg\,=\,\left(\Psi_{\zeta}\circ\mathit{j}_{\mathbb{R}}\right)^{*}\Gg\,=\,\mathit{j}_{\mathbb{R}}^{*}\,\Psi_{\zeta}^{*}\Gg\,=\,\mathit{j}_{\mathbb{R}}^{*}\,\Phi_{\mathbf{U}_{\zeta}}^{*}\Gg\,=\,\mathit{j}_{\mathbb{R}}\Gg\,=\,\Gg^{\mathbb{R}}
\ee
where we used equation \eqref{eqn: relatedness of actions on pure state qubit model} in the fourth equality, and equation \eqref{eqn: unitary invariance of G} in the sixth equality.
Therefore, we conclude that $\Gg^{\mathbb{R}}$ is invariant with respect to the action of the Lie group $G=\mathbb{R}$ on $M=\mathbb{R}$ given by translation, and thus must be proportional to the Euclidean metric tensor.
\end{example}
 
\begin{example}[A mixed state qubit model]\label{ex: a mixed state qubit model} 

Consider the algebra $\mappa_{2}$ of the qubit (see example \ref{ex: the space of density operators}).
Consider the orbit $\mathcal{O}$ of faithful states, set $M=\mathbb{R}^{+}\times\mathbb{R}^{+}$, and define the map $\mathit{j}_{M}$ as
\be\label{eqn: 2-dim qubit mixed model}
\rho_{\gamma,\zeta}\,\equiv\,\mathit{j}_{M}(\gamma,\zeta)\,:=\,\frac{1}{2}\,\left(\sigma_{0} +  \mathrm{e}^{-\zeta\,\gamma} \left(\cos(\gamma)\,\sigma_{1} - \sin(\gamma)\,\sigma_{2} \right)\right).
\ee
A direct computation shows that this map is smooth.
Quite interestingly, the parametric model $(M,\mathit{j}_{M},\mathcal{O})$ has a physical origin which is connected with the dynamics of open quantum systems.
The dynamics of such systems is governed by the so-called Gorini-Kossakowski-Lindblad-Sudarshan (GKLS) equation \cite{AAVV1-2006,AAVV2-2006,AAVV3-2006,C-P-2017,G-K-S-1976,Lindblad-1976}.
In particular, choosing the infinitesimal generator $L$ of this linear equation  to be  the dephasing channel, the dynamical evolution evolution generated by $L$ is such that the initial (pure) state $\rho$ given in equation \eqref{eqn: pure state qubit model R 1} evolves according to the right-hand-side of equation \eqref{eqn: 2-dim qubit mixed model}, where $\frac{\gamma}{2}$ plays the role of  the time parameter while  $2\zeta$ is the dephasing parameter  \cite[ex. 2]{C-DC-I-L-M-2017}.
Note that the initial pure state is evolved into a mixed (faithful) state as soon as the time parameter is greater than 0.

This model has been recently considered in the context of quantum parameter estimation in the presence of nuisance parameters \cite{S-Y-H-2020}.

Let us now consider the vector fields $V$ and $W$ on $M$ generating the local one-parameter groups of local diffeomorphisms
\be
\begin{split}
\phi_{t}(\gamma,\zeta)&\,=\,(\gamma + t ,\zeta) \\
\psi_{t}(\gamma,\zeta)&\,=\,(\gamma ,\zeta + t) .
\end{split}
\ee
Clearly, these vector fields are not complete on $M$, however, they provide a basis of tangent vectors at each point of $M$.
A direct computation shows that
\be
\begin{split}
\langle \mathrm{d}f_{\mathbf{b}},\,T_{\gamma,\zeta}\mathit{j}_{M}(V_{\gamma,\zeta})\rangle &\,=\,-  \mathrm{e}^{-\zeta\gamma} \left(\left(\sin(\gamma)+  \zeta\cos(\gamma)\right)\,b_{1} +  \left(\cos(\gamma) - \zeta\sin(\gamma)\right)\,b_{2}\right)  \\
\langle \mathrm{d}f_{\mathbf{b}},\,T_{\gamma,\zeta}\mathit{j}_{M}(W_{\gamma,\zeta})\rangle \,&\,=\,- \gamma\mathrm{e}^{-\zeta\gamma} \left( \cos(\gamma)\,b_{1} -  \sin(\gamma)\,b_{2}\right),
\end{split}
\ee
from which we conclude that $\mathit{j}_{M}$ is an immersion.
Then, equation \eqref{eqn: SLD} implies that the SLD $\mathbb{Y}_{\mathbf{a}^{V}}(\rho_{\gamma,\zeta})$ and $\mathbb{Y}_{\mathbf{a}^{W}}(\rho_{\gamma,\zeta})$ of $V$ and $W$ at $(\gamma,\zeta)$, respectively, are found as the solutions of
\be
\begin{split}
\langle \mathrm{d}f_{\mathbf{b}},\,T_{\gamma,\zeta}\mathit{j}_{M}(V_{\gamma,\zeta})\rangle &=\langle \mathrm{d}f_{\mathbf{b}},\,\mathbb{Y}_{\mathbf{a}^{V}}(\rho_{\gamma,\zeta})\rangle  = \rho_{\gamma,\zeta}\left(\{\mathbf{a}^{V},\mathbf{b}\}\right) - \rho_{\gamma,\zeta}\left(\mathbf{a}^{V} \right)\,\rho_{\gamma,\zeta}\left(\mathbf{b}\right)\\
\langle \mathrm{d}f_{\mathbf{b}},\,T_{\gamma,\zeta}\mathit{j}_{M}(E_{\gamma,\zeta})\rangle &=\langle \mathrm{d}f_{\mathbf{b}},\,\mathbb{Y}_{\mathbf{a}^{W} }(\rho_{\gamma,\zeta})\rangle = \rho_{\gamma,\zeta}\left(\{\mathbf{a}^{W} ,\mathbf{b}\}\right) - \rho_{\gamma,\zeta}\left(\mathbf{a}^{W} \right)\,\rho_{\gamma,\zeta}\left(\mathbf{b}\right)
\end{split}
\ee
for all self-adjoint elements $\mathbf{b}\in\mappa_{2}$.
A direct computation leads to
\be
\begin{split}
\mathbf{a}^{V} &\,=\, a_{0}^{V} \sigma^{0}-\left(\mathrm{e}^{-\zeta\gamma}\sin(\gamma) + \frac{\zeta\cos(\gamma)}{2\sinh(\zeta\gamma)}\right)\sigma^{1} +  \left( \frac{\zeta \sin(\gamma)}{2\sinh(\zeta\gamma)} - \mathrm{e}^{-\zeta\gamma}\cos(\gamma) \right)\sigma^{2} 
 \\
\mathbf{a}^{W} &\,=\, a_{0}^{W} \sigma^{0}-\frac{\gamma}{2\sinh(\zeta\gamma)}\left(\cos(\gamma)\sigma^{1} - \sin(\gamma)\sigma^{2}\right).
\end{split}
\ee
Note that, apart from the coefficients $ a_{0}^{V} $ and $ a_{0}^{W} $ which
are arbitrary because they do not affect the expression of the associated
gradient vector field, the SLD associated with $V$ and $W$ are   uniquely  defined at each point of $M$.
This is due to the fact that the model is a model of faithful states.
Also, note that $[\mathbf{a}^{V},\mathbf{a}^{W}]\neq\mathbf{0}$, and thus there is no unital, Abelian $C^{*}$-subalgebra of $\mappa_{2}$ that contains both $\mathbf{a}^{V}$ and $\mathbf{a}^{W}$.
This will have an impact on the attainability of the Helstrom bound.

Since $\Gg^{M}=\mathit{j}_{M}^{*}\Gg$, we immediately obtain (see equation \eqref{eqn: Jordan metric tensor})
\be
\Gg^{M}_{\gamma,\zeta}\left(V_{\gamma,\zeta},V_{\gamma\zeta}\right)=\Gg_{\rho_{\gamma,\zeta}}\left(\mathbb{Y}_{\mathbf{a}^{V}}(\rho_{\gamma,\zeta}),\mathbb{Y}_{\mathbf{a}^{V}}(\rho_{\gamma,\zeta})\right)=\rho_{\gamma,\zeta}\left(\{ \mathbf{a}^{V},\mathbf{a}^{V}\}\right) - \left(\rho_{\gamma,\zeta}\left( \mathbf{a}^{V}\right)\right)^{2},
\ee
and similarly for $\Gg^{M}_{\gamma,\zeta}\left(V_{\gamma,\zeta},W_{\gamma\zeta}\right)$ and $\Gg^{M}_{\gamma,\zeta}\left(W_{\gamma,\zeta},W_{\gamma\zeta}\right)$.
Then, since $V$ and $W$ provide a basis of tangent vectors at each point in $M$, the tensor $\Gg^{M}$ can be computed to be
\be
\Gg^{M}\,=\,\left(\mathrm{e}^{-2\zeta\gamma}  + \frac{\zeta^{2}}{\mathrm{e}^{2\zeta\gamma} -1}\right)\mathrm{d}\gamma\otimes\mathrm{d}\gamma + \left(\frac{\zeta\gamma}{\mathrm{e}^{2\zeta\gamma} -1}\right)\mathrm{d}\gamma\otimes_{S}\mathrm{d}\zeta + \left(\frac{\gamma^{2}}{\mathrm{e}^{2\zeta\gamma} -1} \right)\mathrm{d}\zeta\otimes\mathrm{d}\zeta .
\ee
\end{example}

\begin{example}[Lie group and Lie algebra parametric models] 

Motivated by  the model in example \ref{ex: the circle in the space of pure
  states of a qubit}, and by some of the models commonly used in the quantum
context \cite{BN-G-J-2003,Suzuki-2019}, we introduce the notion of a \grit{Lie
  group parametric model} and of a \grit{Lie algebra parametric model}.

Let $G$ be a  Lie group which is realized as a Lie subgroup of the Lie group $\gapp$ of invertible elements in $\appa$, and let $\rho_{0}$ be a state in $\stsp$.
Set $M=G$ and define the   map $\mathit{j}_{G}\colon M \ra \mathcal{O}$, where $\mathcal{O}$ is the orbit containing $\rho_{0}$, by means of
\be\label{eqn Lie group parametric models}
\mathit{j}_{G}(\gr)\,:=\,\Phi(\gr,\rho_{0}).
\ee 
This map is clearly smooth, and we call $(G,\mathit{j}_{G},\mathcal{O})$ a \grit{Lie group parametric model}.
If the fiducial state $\rho_{0}$ is such that
\be\label{eqn: trivial isotropy group for fiducial state of group-adapted parametric models}
\Phi(\gr,\rho_{0})\,=\,\rho_{0} \;\;\Longleftrightarrow\;\;\gr=\mathbb{I}\qquad \forall \gr\in G\,,
\ee
then the model is identifiable.

Since $G$ is a subgroup of $\gapp$, the  left action of $G$ on itself is related with the action of $G$ on $\mathcal{O}$ determined by the restriction of $\Phi$ to $G$ in the way expressed in equation \eqref{eqn: relatedness of actions on pure state qubit model}.
Specifically, let $\psi$ be the  left action of $G$ on itself.
Define an action $\Psi$ of $G$ on $\mathcal{O}$ given by
\be\label{eqn Lie group parametric models 2}
\Psi(g,\rho)\,:=\,\Phi(\gr(g),\rho),
\ee
where $g\in G$ and $\gr(g)\in \gapp$ is the realization of $g$ as an element of $\gapp$.
Then, $g\mapsto\gr(g)$ is a group homomorphism, that is, it satisfies
\be\label{eqn Lie group parametric models 3}
\gr(g_{1}g_{2})\,=\,\gr(g_{1})\,\gr(g_{2}),
\ee
and thus it follows from equations \eqref{eqn Lie group parametric models}, \eqref{eqn Lie group parametric models 2}, and \eqref{eqn Lie group parametric models 3} that
\be
\mathit{j}_{G}\left(\psi(g,h)\right)\,=\,\Psi\left(\gr(g),\mathit{j}_{G}(h)\right),
\ee
which means that the actions $\psi$ and $\Psi$ are equivariant with respect to $\mathit{j}_{G}$.
This means that the fundamental  vector fields of $\psi$ are $\mathit{j}_{G}$-related with the fundamental vector fields of $\Psi$ \cite{A-M-R-1988}.
This instance may be helpful in computing the SLD adapting the steps outlined in example \ref{ex: the circle in the space of pure states of a qubit}.

\vsp

If $(G,\mathit{j}_{G},\mathcal{O})$ is a Lie group parametric model and we consider another parameter manifold which is a smooth homogeneous space $M=G/H$ of $G$ admitting a global, smooth section $\eta\colon M\ra G$, then we can immediately build another parametric model $(M,\mathit{j}_{M},\mathcal{O})$ by setting $\mathit{j}_{M}:=\mathit{j}_{G}\circ\eta$.
This may be helpful to obtain identifiable models.
Indeed, if $\rho_{0}$ has a non-trivial isotropy group $G_{0}\subset G$, which is the set of all elements $g\in G$ such that $\Phi(g,\rho_{0})=\rho_{0}$, we have that $M=G/G_{0}$ is a smooth manifold.
Then, if there is a smooth section $\eta$ for $M$, the resulting parametric model will be identifiable.
This is very similar to the notion of coherent state used in quantum theory \cite{A-A-G-1999,Perelomov-1986}.

\vsp

Another relevant parametric models is obtained when we consider the Lie algebra $\mathfrak{g}$ of $G$.
In this case, we have the exponential map $\exp\colon \mathfrak{g}\ra G$ that can be exploited to define a parametric model.
Specifically, let $(G,\mathit{j},\mathcal{O})$ be a Lie group parametric model.
Then, defining $\mathit{j}_{\mathfrak{g}}:=\mathit{j}_{G}\circ \exp$, we immediately obtain the parametric model $(\mathfrak{g},\mathit{j}_{\mathfrak{g}},\mathcal{O})$ which is referred to as a \grit{Lie algebra parametric model}.
If the Lie algebra $\mathfrak{g}$ is commutative, then  the exponential map is a group homomorphism when the Lie algebra is thought of as a group with respect to the vector sum, and we obtain an equivariance relation with respect to $\mathit{j}_{\mathfrak{g}}$ between the left action $\psi$ of $\mathfrak{g}$ on itself and its realization $\Psi(\mathbf{v},\rho) =\Phi(\exp(\mathbf{v}),\rho)$ as a group acting on $\mathcal{O}$.
\end{example}

\section{Parametric statistical models of states on $C^{\star}$-algebras}\label{sec: statistical models}

When an experiment is performed on a system in a given state $\rho$, we obtain an outcome lying in a given outcome space  $\mathcal{X}$ which is associated with the measurement procedure.
The state $\rho$ is then ``transformed'' into a probability distribution on $\mathcal{X}$ in the sense that different repetitions of the same experimental procedure (i.e., preparation of the system in the state $\rho$ followed by the  measurement procedure with outcome space $\mathcal{X}$) will produce in general different outcomes characterized by a probability distribution which is associated with the state $\rho$ and with the measurement procedure adopted.
In this work, we will always consider outcome spaces which are discrete and finite.

Given a discrete and finite outcome space $\mathcal{X}_{n}$ with $n$ elements, the statistical interpretation of the state $\rho$ is encoded in a map $\mathfrak{m}^{\star}\colon \stsp\lra \mathrm{P}(\mathcal{X}_{n})\equiv\Delta_{n}$, which we will assume to be convex in order to preserve one of the basic features of probabilities and states.
From this, it follows that $\mathfrak{m}^{\star}$ can be  extended to a linear map $\mathfrak{m}^{\star}\colon\appa^{\star}\lra \mathrm{S}(\mathcal{X}_{n})$, where $\mathrm{S}(\mathcal{X})$ is the vector space of signed measures on $\mathcal{X}_{n}$.
From the $C^{\star}$-algebraic perspective, $\mathrm{S}(\mathcal{X}_{n})$ is the space of self-adjoint linear functionals on the Abelian $C^{\star}$-algebra $\cappa_{n}:=C(\mathcal{X}_{n})$ of complex-valued, continuous functions on $\mathcal{X}_{n}$, and thus, since $\mathfrak{m}^{\star}$ is continuous because $\appa^{\star}$ and $\mathrm{S}(\mathcal{X}_{n})$ are finite-dimensional, we immediately obtain that there is a continuous linear map $\mathfrak{m} \colon \cappa_{n}\lra\appa$ of which $\mathfrak{m}^{\star}$ is the dual map.
By construction, the  map $\mathfrak{m}$ must  be such that its dual map $\mathfrak{m}^{\star}$ sends the space of states of $\appa$  into the space of states of $\cappa_{n}$.
One way to implement this condition is to require $\mathfrak{m} \colon\cappa_{n} \lra\appa$ to  be a unital, positive map between $C^{\star}$-algebras, that is, a linear map preserving the identity and sending positive elements into positive elements (clearly, any such map sends self-adjoint elements into self-adjoint elements).

\begin{definition}\label{def: measurement procedure}
A positive unital map $\mathfrak{m}\colon\cappa_{n}\ra\appa$  is defined to be a \grit{measurement procedure}.
\end{definition}

Specifically, given a finite and discrete outcome space $\mathcal{X}_{n}$, we can always consider the basis of $\cappa_{n}$ given by the elements $\{\mathbf{e}^{j}\}_{j=1,...,n}$ where  $\mathbf{e}^{j}$ is the ``delta function'' at the $j$-th element of $\mathcal{X}_{n}$.
The measurement procedure $\mathfrak{m}$ amounts to define the elements
\be
\mathbf{m}^{j}\,:=\,\mathfrak{m}(\mathbf{e}^{j})\qquad \forall\,j=1,...,n ,
\ee
in such a way that they satisfy 
\be\label{eqn: unitality of measurement procedure}
\sum_{j=1}^{n}\,\mathbf{m}^{j}\,=\,\mathbb{I},
\ee
and 
\be\label{eqn: positivity of measurement procedure}
\mathbf{m}^{j}\,\geq\,\mathbf{0}\qquad \forall\,j=1,...,n .
\ee
Essentially, we are considering a (discrete) POVM in the $C^{\star}$-algebraic framework.
The probability distribution $\mathfrak{m}^{\star}(\rho)$ associated with the state $\rho$ is characterized by the numbers
\be\label{eqn: probability distribution associated with a state by a measurement procedure}
p^{j}\,:=\,\left(\mathfrak{m}^{\star}(\rho)\right)(\mathbf{e}^{j})\,=\,\rho\left(\mathfrak{m}(\mathbf{e}^{j})\right)\,=\,\rho(\mathbf{m}^{j})\,.
\ee

Once a parametric model $(M,\mathit{j},\mathcal{O})$ is chosen, we immediately have the map
\be\label{eqn: classical immersion}
\mathit{j}^{c}\,:=\,\mathfrak{m}^{\star}\circ \mathit{j}\colon M\,\lra\, \Delta_{n}\,.
\ee
We require this map to lie entirely in a given fixed orbit of states inside $\Delta_{n}$.
Clearly, since every orbit in $\Delta_{n}$ is diffeomorphic to $\Delta_{k}^{+}$ for some $k\neq n$ (see example \ref{ex: the probability simplex}), there is no loss of generality in requiring the codomain of $\mathit{j}^{c}$ to lie entirely inside the manifold $\Delta_{n}^{+}$ of faithful states on $\cappa_{n}$.
Indeed, if this is not the case, it suffices to redefine $\mathcal{X}_{n}$ to be the subset $I_{k}$,  exchange $\cappa_{n}$ with $C(I_{k})$, and relabel $k$ as $n$.

\begin{definition}\label{def: regular measurement procedure}
Let $(M,\mathit{j},\mathcal{O})$ be a parametric model of states on a $C^{\star}$-algebra $\appa$.
A  \grit{measurement procedure} $\mathfrak{m}$ such that $\mathit{j}^{c}(M)\,:=\,\mathfrak{m}^{\star}\circ \mathit{j}(M)\subseteq\Delta_{n}^{+}$ is called \grit{regular} for $(M,\mathit{j},\mathcal{O})$.
\end{definition}

Once a regular measurement procedure $\mathfrak{m}$ for $(M,\mathit{j},\mathcal{O})$ is chosen, we are ready to build a parametric statistical model (in the sense of information geometry \cite{AAVV1-1987,A-N-2000,A-J-L-S-2017}) associated with the parametric model $(M,\mathit{j},\mathcal{O})$.

\begin{definition}\label{def: parametric statistical model}
Let $(M,\mathit{j},\mathcal{O})$ be a parametric model of states on a $C^{\star}$-algebra $\appa$, and let $\mathfrak{m}$ be a regular measurement procedure for $(M,\mathit{j},\mathcal{O})$.
Then, the triple $(M,\mathit{j}^{c},\Delta_{n}^{+})$, with $\mathit{j}^{c}$ as in equation \eqref{eqn: classical immersion}, is defined to be the \grit{parametric statistical model} associated with the parametric model $(M,\mathit{j},\mathcal{O})$ by means of the  measurement procedure $\mathfrak{m}$.
\end{definition}

 
 The  open interior of the simplex $\Delta_{n}^{+}$ coincides with the space of faithful states on the finite-dimensional, commutative $C^{\star}$-algebra $\cappa_{n}$ of complex-valued continuous functions on the discrete n-point space $\mathcal{X}{n}$, and thus  the Radon-Nikodym derivative of  $\mathbf{p}\in\Delta_{n}^{+}$ with respect to the counting measure on $\mathcal{X}_{n}$ is well-defined as a function on $\mathcal{X}_{n}$ and it is called the probability density function of $\mathbf{p}$.
Clearly, being $\mathit{j}^{c}(M)\subseteq\Delta_{n}^{+}$, every element $m\in M$ may be uniquely associated with the probability density function of $\mathbf{p}_{m}=\mathit{j}^{c}(m)$.
Moreover, for every $x\in\mathcal{X}_{n}$, the function $p(m,x)=\mathbf{p}_{m}(\{x\})$ is a smooth function on $M$ because  $\mathbf{p}_{m}$ is a linear functional on $\cappa_{n}$ and $\mathit{j}^{c}$ is smooth, and its support does not depend on the chosen $x\in\mathcal{X}_{n}$ because $\mathit{j}^{c}(M)\subseteq\Delta_{n}^{+}$.
These regularity properties are particularly meaningful with respect to the Cramer-Rao bound discussed in section \ref{sec: cramer-rao bound for manifold-valued estimators}. 


\begin{remark}[Classical statistical models]\label{ex: classical statistical models}

In the specific case when the algebra $\appa$ is commutative, i.e., $\appa=\cappa_{n}$ for some $n\in\mathbb{N}$, a \grit{parametric model} $(M,\mathit{j},\mathcal{O})$  of states on $\cappa_{n}$ is already a \grit{parametric statistical model} by itself.
Indeed, according to example \ref{ex: the probability simplex}, the orbit $\mathcal{O}$ is diffeomorphic to the open interior $\Delta_{k}^{+}$ of a  k-simplex with $k\neq n$.
Specifically, we have a subset $I_{k}\subseteq\mathcal{X}_{n}$ of $k$ elements,  the $C^{\star}$-algebra $\cappa_{k}$ generated by the elements $\mathbf{e}^{j}\in\cappa_{n}$ with $j$ such that $\mathbf{x}_{j}\in I_{k}$, and $\mathcal{O}$ is diffeomorphic to the orbit of faithful states of $\cappa_{k}$.
Then, we have a ``natural'' measurement procedure $\mathfrak{m}\colon\cappa_{k}\ra\cappa_{n}$ at our disposal given by the natural identification $\mathrm{i}_{k}$ map of $\cappa_{k}$ in $\cappa_{n}$,  and the map $\mathit{j}^{c}=\mathfrak{m}^{\star}\circ \mathit{j}=\mathrm{i}_{k}^{*}\circ\mathrm{j}$ gives rise to the statistical model $(M,\mathit{j}^{c},\Delta_{k}^{+})$ associated with $(M,\mathit{j},\mathcal{O})$.
From this, it is clear that once we have the parametric model $(M,\mathit{j},\mathcal{O})$ we immediately have a ``natural'' parametric statistical model $(M,\mathit{j}^{c},\Delta_{k}^{+})$ associated with it.
No additional choices must be made. 

\end{remark}
 

Exploiting the Riemannian geometry of $\Delta_{n}^{+}$, the parameter manifold $M$ may be endowed with another symmetric, covariant $(0,2)$ tensor which is in general different from the metric $\Gg^{M}$ introduced before.
Indeed, we may consider the Fisher-Rao Riemannian metric $\Gg_{FR}$ on $\Delta_{n}^{+}$, which is the Riemannian metric tensor $\Gg$ associated with the Jordan product of the self-adjoint part of $\cappa_{n}$ as described in section \ref{sec: differential geometry of states}, and then take its pullback 
\be\label{eqn: classical pullback metric}
\Gg^{Mc}=(\mathit{j}^{c})^{*}\Gg_{FR}
\ee
to $M$ (the `c' stands for classical, or commutative).
In this case, we obtain a symmetric covariant tensor on $M$ which, unlike   $\Gg^{M}$ given by equation \eqref{eqn: pullback metric}, can not feel the possible non-commutativity of $\appa$, and which is  the pullback of the Fisher-Rao metric tensor on $M$ thought of as a parametric statistical model in $\Delta_{n}^{+}$ along the lines of classical information geometry.

\vsp

To accomodate multiple runs, say $N$, of the same experimental procedure  on  $N$ identical and independent copies of the initial state, we introduce the parametric model $(M,\mathit{j}^{N},\mathcal{O}^{N})$ where $\mathcal{O}^{N}$ is the manifold of states on the tensor product algebra 
\be
\appa^{\otimes N}:=\appa\otimes\cdots\otimes\appa
\ee
containing the product states of the form $\rho_{1}\otimes\cdots\otimes\rho_{N}$ with $\rho_{j}\in\mathcal{O}$ for every $j=1,...,N$, and $\mathit{j}^{N}\colon M\lra\mathcal{O}^{N}$ is given by
\be
\mathit{j}^{N}(m)\,:=\,\mathit{j}(m)\,\otimes\,\cdots\,\otimes\,\mathit{j}(m)\,\equiv\,\rho_{m}\,\otimes\,\cdots\,\otimes\,\rho_{m}\,\equiv\,\rho_{m}^{\otimes N}.
\ee
Clearly, we may endow $M$ with the Riemannian metric $\Gg^{MN}$ defined by
\be
\Gg^{MN}\,:=\,(\mathit{j}^{N})^{*}\Gg^{N},
\ee
where $\Gg^{N}$ denotes the canonical Riemannian metric on $\mathcal{O}^{N}$ associated with the Jordan product on $\appa^{ \otimes N}$.
Since the smooth embedding $\mathit{j}^{N}$  has been defined  in terms of a ``multiplicative object'', namely, the tensor product, it is reasonable to expect that this multiplicative feature reflects also in the pullback metric.
Indeed, below we will prove   that
\be
\Gg^{MN}\,=\,N\Gg^{M}.
\ee

Performing $N$ runs of an experiment  provides us with a list of $N$ outcomes, and  we consider the outcome space 
\be
\mathcal{X}^{N}\,=\,\mathcal{X}\,\times\,\cdots\,\times\,\mathcal{X} .
\ee
At this point, we must choose a  measurement procedure   $\mathfrak{m}^{N}\colon \cappa_{n}^{ \otimes N}=C(\mathcal{X}^{N})\lra\appa^{\otimes N}$ so that, setting $\mathit{j}^{cN}=\mathfrak{m}^{N}\circ\mathit{j}^{N}$,  we can build a statistical   model $(M,\mathit{j}^{cN},\Delta_{Nn}^{+})$ in the obvious way.
We may  endow $M$ with the Riemannian metric $\Gg^{McN}$ defined by
\be
\Gg^{McN}\,:=\,N(\mathit{j}^{cN})^{*}\Gg_{FR},
\ee
where $N\Gg_{FR}$ is the Fisher-Rao metric tensor on $\Delta_{nN}^{+}$ (this either follows from standard arguments in classical information geometry, or by proposition \ref{prop: Jordan metric for N rounds is N times the Jordan metric}  below applied to the case where $\appa=\cappa_{n}$).

\begin{proposition}\label{prop: Jordan metric for N rounds is N times the Jordan metric}

With the notations introduced above, we have 
\be
\Gg^{MN}\,=\,N\Gg^{M}.
\ee

\end{proposition}

\begin{proof}

We start proving that, if $v_{m}\in T_{m}M$ is such that 
\be
T_{m}\mathit{j}(v_{m})\,=\,\mathbb{Y}_{\mathbf{a}}(\rho_{m}),
\ee
then it holds
\be\label{eqn: iN-related gradient tangent vector}
T_{m}\mathit{j}^{N}(v_{m})\,=\,\mathbb{Y}_{a^{ N}}^{N}(\rho_{m}^{\otimes N}),
\ee
where $\mathbb{Y}_{a^{N}}^{N}$ is the gradient vector field on $\mathcal{O}^{N}$ associated with 
\be\label{eqn: special product elements}
\mathbf{a}^{N}\,=\,\mathbf{a}\,\otimes\,\mathbb{I}\,\otimes\,\cdots\,\otimes\,\mathbb{I} + \mathbb{I}\,\otimes\,\mathbf{a}\,\otimes\,\mathbb{I}\,\otimes\,\cdots\,\otimes\,\mathbb{I} + \cdots + \mathbb{I}\,\otimes\,\cdots\,\otimes\,\mathbb{I}\,\otimes\,\mathbf{a}.
\ee
Recall that simple elements of the form $\mathbf{b}_{1}\,\otimes\cdots\,\otimes\,\mathbf{b}_{N}$ generate $\appa^{\otimes N}$, and thus, to prove equation \eqref{eqn: iN-related gradient tangent vector}, it is sufficient to compute
\be
\langle \mathrm{d}f_{\mathbf{b}_{1}\otimes\cdots\otimes\mathbf{b}_{N}}(\rho_{m}^{\otimes N}),\, T_{m}\mathit{j}^{N}(v_{m})\rangle\,=\,\langle \mathrm{d}(\mathit{j}^{N})^{*}f_{\mathbf{b}_{1}\otimes\cdots\otimes\mathbf{b}_{N}}(m),\, v_{m}\rangle .
\ee
Denoting by $m_{t}$ a smooth curve in $M$ starting at $m$ with initial tangent vector $v_{m}$, we have
\be
\begin{split}
\langle \mathrm{d}(\mathit{j}^{N})^{*}f_{\mathbf{b}_{1}\otimes\cdots\otimes\mathbf{b}_{N}}(m),\, v_{m}\rangle &\,=\,\frac{\mathrm{d}}{\mathrm{d}t}\,\left(\rho_{m_{t}}^{\otimes N}(\mathbf{b}_{1}\,\otimes\,\cdots\,\otimes\,\mathbf{b}_{N})\right)_{t=0}\,=\, \\ 
&\,=\,\frac{\mathrm{d}}{\mathrm{d}t}\,\left(\rho_{m_{t}}(\mathbf{b}_{1})\,\cdots\,\rho_{m_{t}}(\mathbf{b}_{N})\right)_{t=0},  
\end{split}
\ee
from which equation \eqref{eqn: iN-related gradient tangent vector} follows applying the Leibniz rule and recalling that $T_{m}\mathit{j}(v_{m})\,=\,\mathbb{Y}_{\mathbf{a}}(\rho_{m})$.

We now take $v_{m},w_{m}\in T_{m}M$ such that 
\be
\begin{split}
T_{m}\mathit{j}^{N}(v_{m})&\,=\,\mathbb{Y}_{a^{ N}}^{N}(\rho_{m}^{\otimes N}) \\
T_{m}\mathit{j}^{N}(w_{m})&\,=\,\mathbb{Y}_{b^{ N}}^{N}(\rho_{m}^{\otimes N}),
\end{split}
\ee
with $\mathbf{a}^{N}$ and $\mathbf{b}^{N}$ as in equation \eqref{eqn: special product elements}.
Recalling that $\Gg^{MN}=(\mathit{j}^{N})^{*}\Gg^{N}$,  and noting that
\be
\Gg^{N}_{\rho_{m}^{\otimes N}}(\mathbb{Y}_{a^{N}}^{N}(\rho_{m}^{\otimes N}),\,\mathbb{Y}_{b^{N}}^{N}(\rho_{m}^{\otimes N}))\,=\,\rho_{m}^{\otimes N}(\{\mathbf{a}^{N},\mathbf{b}^{N}\}) - \rho_{m}^{\otimes N}(\mathbf{a}^{N})\,\rho_{m}^{\otimes N}(\mathbf{b}^{N})
\ee 
because of equation \eqref{eqn: Jordan metric tensor},  we have
\be
\begin{split}
\Gg^{MN}_{m}(v_{m},w_{m})&\,=\,\Gg^{N}_{\rho_{m}^{\otimes N}}(\mathbb{Y}_{a^{N}}^{N}(\rho_{m}^{\otimes N}),\,\mathbb{Y}_{b^{N}}^{N}(\rho_{m}^{\otimes N}))\\
&\,=\, \rho_{m}^{\otimes N}\left(\{\mathbf{a}^{N},\,\mathbf{b}^{N}\}\right) - \rho_{m}^{\otimes N}(\mathbf{a}^{N})\,\rho_{m}^{\otimes N}(\mathbf{b}^{N})\,=\, \\
&\,=\,\left(N\rho_{m}(\{\mathbf{a},\mathbf{b}\}) + N(N-1)\,\rho_{m}(\mathbf{a})\,\rho_{m}(\mathbf{b})\right) - N^{2}\,\rho_{m}(\mathbf{a})\,\rho_{m}(\mathbf{b})\,=\, \\
&\,=\, N\left(\rho_{m}(\{\mathbf{a},\mathbf{b}\}) - \rho_{m}(\mathbf{a})\,\rho_{m}(\mathbf{b})\right)\,=\, \\
&\,=\, N\,\Gg^{M}_{m}(v_{m},\,w_{m})\,
\end{split}
\ee
as desired.

\end{proof}

\section{The problem of estimation theory}\label{sec: estimation theory}

The purpose of estimation theory is to manipulate the outcomes of experiments in such a way to obtain an estimate of the ``true state'' on which the experiment has been performed.
This is done by means of a map $\mathcal{E}\colon\mathcal{X}_{n} \lra\,M$ called \grit{estimator}.
In the following, we will always consider \grit{non-constant} estimators.

Clearly, we need to come up with a way of establishing optimality for estimators.
For this purpose, we introduce a smooth  \grit{cost function} $C\colon M\times M\lra \mathbb{R}$ which is non-negative and vanishes only on the diagonal.
The choice of the cost function is essentially left to the ingenuity of the theoretician, and it is difficult to outline a general selection methodology.
However, in some cases, the choice of the cost function is suggested by the context. 

Starting with a cost function $C$, and writing $\mathcal{E}_{j}\equiv\mathcal{E}(\mathbf{x}_{j})$  for the value of the estimator  at the $j$-th element of the outcome space $\mathcal{X}_{n}$, we introduce the function $L\colon M\times M\lra \mathbb{R}$ given by
\be\label{eqn: parametric loss functional}
L(m_{1},m_{2})\,:=\,\sum_{j=1}^{n}\,C(m_{1},\mathcal{E}_{j})\,p^{j}(\mathrm{m}_{2})\,=\,\sum_{j=1}^{n}\,C(m_{1},\mathcal{E}_{j})\,\rho_{m_{2}}(\mathbf{m}^{j}),
\ee
where $(p^{1}(m_{2}),\cdots,p^{n}(m_{2}))=\mathit{j}^{c}(m_{2})\,=\mathfrak{m}^{\star}(\rho_{m_{2}})$, and $\mathfrak{m}$ is the measurement procedure ``generating'' the statistical   model $(M,\mathit{j}^{c},\Delta_{n}^{+})$ associated with the parametric model $(M,\mathit{j},\mathcal{O})$ of states on $\appa$ under investigation.
It is clear from equation \eqref{eqn: parametric loss functional} that if the cost function $C$ is constant, then $L$ does not actually depend on $m_{2}$, and the problem of estimation theory as will be now developed will lose meaning.

The function $L$ may be seen as the expectation value of the real-valued, $M$-parametric random variable $C(m_{1},\mathcal{E}(\cdot))$ on $\mathcal{X}_{n}$ with respect to the $M$-parametric probability distribution $\mathfrak{m}(\rho_{m_{2}})$ on $\mathcal{X}_{n}$.
Therefore, $L$ measures  how centered is the probability distribution generated by   $C(m_{1},\mathcal{E}(\cdot))$.

Let $m_{\star}\in M$ and denote by $L_{\star}$ the function
\be\label{eqn: parametric loss functional 2}
L_{\star}(m)\,:=\,L(m,m_{\star}).
\ee
The estimator $\mathcal{E}$ is called \grit{stationary} for the cost function $C$ at $m_{\star}$  if $L_{\star}$ has an extremum at $m=m_{\star}$, that is, if 
\be\label{eqn: stationary estimator}
\left(VL_{\star}\right)(m_{\star})\,=\,0
\ee
for all vector fields $V$ on $M$.
The estimator $\mathcal{E}$ is called \grit{unbiased} for the cost function $C$ at  $m_{\star}\in M$ if the function $L_{\star}$ has a  minimum at $m=m_{\star}$, and it is called \grit{locally unbiased} for the cost function $C$ at $m_{\star}$ if $L_{\star}$ has a   local minimum at $m=m_{\star}$.
In general, for a given  cost function $C$, unbiased estimators need not exist.

Now, we may define an $M$-parametric self-adjoint element $\mathcal{M}$ in $\appa$ setting
\be
\mathcal{M}_{m_{1}}\,:=\,\sum_{j=1}^{n}\,C(m_{1},\mathcal{E}_{j})\, \mathbf{m}^{j}\,.
\ee
This element clearly depends also on the estimator $\mathcal{E}$ and on the measurement procedure $\mathfrak{m}$.
Moreover, it allows us to write the function $L$ as the expectation value of $\mathcal{M}_{m_{1}}$ with respect to the state $\rho_{m_{2}}$ according to
\be
L(m_{1},m_{2})\,=\,\rho_{m_{2}}\left(\mathcal{M}_{m_{1}}\right).
\ee 

The estimation problem may be approached from two different perspectives of increasing difficulty:
\begin{itemize}
\item the regular measurement procedure $\mathfrak{m}$ is fixed, and the unknown of the problem is the estimator $\mathcal{E}$;
\item both the regular measurement procedure $\mathfrak{m}$ and the estimator $\mathcal{E}$ are considered unknown.
\end{itemize}
Clearly, the first case reduces to the classical problem of estimation, and may be faced relying on well-known methods like the maximum likelyhood estimator.
The limit on the precision is then governed by the Cramer-Rao bound (see section \ref{sec: cramer-rao bound for manifold-valued estimators}).
The second case is definitely more difficult to address because the freedom in the choice of the regular measurement procedure adds another layer of complexity.
However, in this case, the precision is governed by the Helstrom bound  (see section \ref{subsec: helstrom bound for manifold-valued estimators}), and allows for a sharpening of the Cramer-Rao bound.
Indeed, the freedom in choosing the measurement procedure reflects in the possibility of consider different ``classical scenarios'', and choose the one with the lowest Cramer-Rao bound.

Unfortunately, for both forms of the problem, there is no algorithm to solve the problem in full generality, and a case-by-case analysis is mandatory.

\begin{remark}[Stationary estimators for Euclidean cost function]\label{subsec: stationary estimators for Euclidean cost function}
 
Suppose that $M$ is explicitely realized as an $n$-dimensional submanifold of $\mathbb{R}^{N}$ for some positive $N\in\mathbb{N}$ with $n\leq N$.
In this context, a common choice in parameter estimation theory is to consider  the  cost function $C$ which is the   Euclidean distance on $\mathbb{R}^{N}\times\mathbb{R}^{N}$ restricted to $M\times M$.
Specifically, we have
\be
C(m_{1},m_{2})\,:=\,\frac{1}{2}\,\left| m_{1} - m_{2}\right|^{2}\,,
\ee
so that the function $L$ reads
\be
L(m_{1},m_{2})\,:=\,\frac{1}{2}\,\sum_{j=1}^{n}\,\left| m_{1} - \mathcal{E}_{j}\right|^{2}\,p^{j}(\mathrm{m}_{2}).
\ee
This type of cost function is called a \grit{Euclidean cost function} for obvious reasons.
Clearly, the Euclidean cost function $C$ depends on the actual realization of the (a priori abstract) manifold $M$ into a suitable $\mathbb{R}^{N}$.
In particular, because of Whitney's embedding theorem, given a parameter manifold $M$ we can always build a Euclidean cost function.
Of course, the actual usefulness of such a cost function is in principle not clear and should be investigated case by case.
However, it often happens in concrete models that the parameter manifold $M$ is ``naturally'' immersed in some given $\mathbb{R}^{N}$ by construction, and thus the Euclidean cost function unavoidably presents itself from the start.

If $\{\theta^{1},...,\theta^{n}\}$ is a local system of coordinates on $M$, it is easy to see that being stationary at $m_{\star}$ is equivalent to (see equation \eqref{eqn: stationary estimator})
\be\label{eqn: stationary estimator for open submanifold of Euclidean space}
m_{\star}^{k}(\theta) = \mathrm{E}_{m_{\star}(\theta)}[\mathcal{E}^{k}] \qquad \forall k=1,...,N \mbox{ and } \;\forall r=1,...,n,
\ee
where $m_{1}^{k}$ is the smooth function on $M$ obtained by composing the canonical immersion of $M$ in $\mathbb{R}^{N}$ with the canonical projection on the $k$-th factor, $\mathcal{E}^{k}$ is the real-valued random variable on $\mathcal{X}$ obtained by composing $\mathcal{E}$ with the canonical immersion of $M$ in $\mathbb{R}^{N}$ and with the canonical projection on the $k$-th factor, and where $\mathrm{E}_{m_{\star}}[\cdot]$ denote the expectation value with respect to the $M$-parametric probability distribution $\mathbf{p}_{m_{\star}}$.

Since $C > 0$ for all $(m_{1},m_{2})\in M\times M$ unless $m_{1}=m_{2}$, in which case it vanishes, we see that a stationary estimator at $m\in M$ is also locally unbiased at $m\in M$.

When $M$ is an open subset of $\mathbb{R}^{N}$ and $\{\theta^{1},...,\theta^{N}\}$ is a system of Cartesian coordinates, and when equation \eqref{eqn: stationary estimator for open submanifold of Euclidean space} holds for all $m\in M$,  we recover the  standard  definition of an unbiased estimator used in classical and quantum estimation theory \cite[ch. 4]{A-N-2000}.

\end{remark}
 
\section{The Cramer-Rao bound}\label{sec: cramer-rao bound for manifold-valued estimators}

Here, we recall Hendrik's derivation of the Cramer-Rao bound for estimators with values in a manifold \cite{Hendriks-1991} when the underlying outcome space is discrete and finite.
This gives a clear geometric picture of the Cramer-Rao  bound which does not depend on the existence of a privileged coordinatization of the parameter space $M$ as it is the case in most of the existing literature (see for instance  \cite[ch. 4]{A-N-2000} where it is clearly stated that the notion of unbiased estimator developed there is coordinate-dependent, as well as  \cite{G-O-2006,Hendriks-1991,O-C-1995})


\vsp

Let $(M,\mathit{j},\Delta_{n}^{+})$ be a  parametric statistical model.
We refer to definition \ref{def: parametric statistical model} and the paragraph right after it, as well as to  remark \ref{ex: classical statistical models} for a discussion of the regularity properties satisfied by the model $(M,\mathit{j},\Delta_{n}^{+})$.
Recall that the metric $\Gg^{M}$ determined by equation \eqref{eqn: pullback metric} coincides with the   Fisher-Rao  tensor on $M$ as determined by standard methods of information geometry \cite{AAVV1-1987,Amari-2016,A-N-2000}.
We assume that $\Gg^{M}$ is invertible.

In order to obtain the generalized Cramer-Rao bound for a stationary estimator, we need to exploit the geometrical properties of the product structure of the manifold $M\times M$.
We will now recall these geometrical properties following \cite[sec. 2]{C-DC-L-M-M-V-V-2018}, to which we refer for the explicit proofs.

First of all, we note that there are two projections $\pi_{l}$ and $\pi_{r}$ from $M\times M$ to $M$ given by
\be
\begin{split}	
\pi_{l}(m_{1},m_{2})&\,:=\,m_{1} \\
\pi_{r}(m_{1},m_{2})&\,:=\,m_{2},
\end{split}
\ee
and there is also the diagonal immersion $i_{d}$ of  $M$ into $M\times M$ given by
\be
\mathrm{i}_{d}(m)\,:=\,(m,\,m)\,.
\ee
Given a vector field $X$ on $M$, we may define its left and right lift to be the vector fields $X_{l}$ and $X_{r}$ on $M\times M$ characterized by
\be
\begin{split}
X_{l}(\pi_{l}^{*}f)&\,=\,\pi_{l}^{*}(X(f)) \\
X_{r}(\pi_{r}^{*}f)&\,=\,\pi_{r}^{*}(X(f)) 
\end{split}
\ee
for every smooth function $f$ on $M$.
It is possible to prove that every vector field $X$ on $M$ is $i_{d}$-related with the vector field $X_{l}+X_{r}$ on $M \times M$ \cite[sec. 2]{C-DC-L-M-M-V-V-2018}.


If $\mathcal{E}$ is a stationary estimator at $m_{\star}\in M$ then $L_{\star}$ has an extremum at $m_{\star}$, and this is equivalent to 
\be
\left(\mathrm{i}_{d} \left(X_{l} L \right)\right)_{m=m_{\star}}\,=\,0
\ee
for all vector fields $X_{l}$ on $M\times M$.
We assume that $\mathcal{E}$ is a stationary estimator for all $m_{\star}\in M$.
This means that the function $\mathrm{L}=\mathrm{i}_{d} \left(X_{l} L \right)$ identically vanishes.
Consequently, given an arbitrary vector field $Y$ on $M$, we also have
\be
0\,=\, Y\left(i_{d}^{*}\mathrm{L}\right) \,=\, Y\left(i_{d}^{*}\left(X_{l} L\right)\right) \,=\, i_{d}^{*}\left((Y_{l}X_{l} + Y_{r} X_{l}) L\right),
\ee
which means
\be\label{eqn: for the C-R 1}
i_{d}^{*}\left(Y_{l}X_{l}  L\right) \,=\,- i_{d}^{*}\left(Y_{r} X_{l} L\right)\,.
\ee

Since $\mathcal{E}$ is stationary at every $m_{\star}$, it follows that  the Hessian form $\mathrm{H}_{\star}$ of $L_{\star}$ at $m_{\star}$ is well defined and we have 
\be
\mathrm{H}_{\star}(X(m_{\star}),\,Y(m_{\star}))\,:=\,\left(Y\,X\,L_{\star}\right)(m_{\star}).
\ee
A moment of reflection shows that
\be
\left(Y\,X\,L_{\star}\right)(m_{\star})\,=\,\left(i_{d}^{\star}\left(Y_{l}X_{l}  L\right)\right)(m_{\star})
\ee
so that
\be
\mathrm{H}_{\star}(X(m_{\star}),\,Y(m_{\star}))\,=\,-\left(i_{d}^{\star}\left(Y_{r}X_{l}  L\right)\right)(m_{\star}) 
\ee
because of equation \eqref{eqn: for the C-R 1}.
Set
\be
\begin{split}
C_{\mathcal{E}_{j}}(m)&\,:=\,C(m,\mathcal{E}_{j}) \\
\end{split}
\ee
so that we have
\be
L(m_{1},m_{2})\,:=\,\sum_{j=1}^{n}\,C_{\mathcal{E}_{j}}(m_{1})\,p^{j}(\mathrm{m}_{2})
\ee
and we obtain
\be\label{eqn: almost cramer rao}
\mathrm{H}_{\star}(X(m_{\star}),\,Y(m_{\star}))\,=\, -\sum_{j=1}^{n}\,\left(X_{l}C_{\mathcal{E}_{j}}\right)(m_{\star})\, \left(Y_{r}p^{j}\right)(m_{\star}) \,.
\ee
Introducing the real-valued random variables on the probability space $(\mathcal{X}_{n},\mathbf{p}(m_{\star}))$ given by
\be\label{eqn: C-R 01}
\begin{split}
F_{X}^{\star}(\mathbf{x}_{j})&\,:=\,\left(X_{l}C_{\mathcal{E}_{j}}\right)(m_{\star})\\
G_{Y}^{\star}(\mathbf{x}_{j})&\,:=\,\left(Y_{r}\ln(p^{j})\right)(m_{\star}),
\end{split}
\ee
we can rewrite the right hand side of equation \eqref{eqn: almost cramer rao} as
\be
\mathrm{H}_{\star}(X(m_{\star}),\,Y(m_{\star}))\,=\,-  \mathrm{E}_{\star}\left[F_{X}^{\star}\,G_{Y}^{ \star}\right]\,,
\ee
where $\mathrm{E}_{\star}\left[\cdot\right]$ denotes the expectation value with respect to the probability measure $\mathbf{p}(m_{\star})$. 
The expression 
\be
\langle F,G\rangle_{ \star }\,:=\,\mathrm{E}_{ \star }\left[F\,G\right]
\ee
is an inner product on the space of random variables on the probability space $(\mathcal{X}_{n},\mathbf{p}(m_{\star})$, and the Cauchy-Schwarz inequality may be applied to obtain
\be
\left(\mathrm{H}_{ \star }(X(m_{\star}),\,Y(m_{\star}))\right)^{2}\,\leq\,\mathrm{E}_{ \star  }\left[F_{X}^{ \star }\,F_{X}^{ \star }\right]\,\mathrm{E}_{ \star }\left[G_{Y}^{ \star} \,G_{Y}^{ \star }\right]\,.
\ee
Then,  a direct computation shows that 
\be
\begin{split}
\mathrm{E}_{ \star }\left[G_{Y}^{ \star }\,G_{Y}^{ \star }\right]&\,=\,\sum_{j=1}^{n}\left(Y_{r}\ln(p^{j})\right)(m_{\star})\left(Y_{r}\ln(p^{j})\right)(m_{\star})\,p^{j}(m_{\star})\,=\,\\
&\,=\,\Gg^{M}(Y(m_{\star}),\,Y(m_{\star})).
\end{split}
\ee

Next, we introduce the expression 
\be\label{eqn: C-R 02}
\begin{split}
\mathcal{C}(X(m_{\star}),Y(m_{\star}))&\,:=\,\mathrm{E}_{\star}\left[F_{X}^{ \star }\,F_{Y}^{ \star}\right],
\end{split}
\ee
which according to \eqref{eqn: C-R 01} implicitly contains the cost function $C$, so that we can write equation \eqref{eqn: almost cramer rao} as
\be\label{eqn: almost cramer rao 2}
\left(\mathrm{H}_{ \star }(X(m_{\star}),\,Y(m_{\star}))\right)^{2}\,\leq\,\mathcal{C}(X(m_{\star}),X(m_{\star}))\;\Gg^{M}(Y(m_{\star}),\,Y(m_{\star}))\,.
\ee
Clearly, $\mathcal{C}$ depends on the cost function $C$ and the estimator $\mathcal{E}$.

Now,  fix $X_{m_{\star}}\in T_{m_{\star}}M$, and define the function $H\colon T_{m_{\star}}M\lra\mathbb{R}$ given by
\be\label{eqn: for C-R-2}
Y(m_{\star})\equiv Y_{m_{\star}}\,\mapsto\,H(Y_{m_{\star}})\,:=\, \mathrm{H}_{ \star }(X_{m_{\star}},\,Y_{m_{\star}}).
\ee
This function admits a maximum on the unit sphere determined by the Fisher-Rao metric.
Indeed, the Fisher-Rao unit sphere in $T_{m_{\star}}M$ is compact because the Fisher-Rao metric is  a Riemannian metric (positive).
Let $Y^{0}_{m_{\star}}$ be a point on which $H$ is maximum.
Then, we may always find a real number $\lambda$ such that
\be\label{eqn: for C-R-3}
H(Y_{m_{\star}})\,=\, \lambda\, \Gg^{M}(Y^{0}_{m_{\star}},Y_{m_{\star}})\,,
\ee 
so that 
\be\label{eqn: for C-R-4}
H(Y_{m_{\star}}^{0})\,=\, \lambda\, \Gg^{M}(Y^{0}_{m_{\star}},Y_{m_{\star}}^{0})\,=\, \lambda
\ee
because $Y^{0}_{m_{\star}}$ lies on the Fisher-Rao unit sphere.

With an evident abuse of notation, we denote by $\mathrm{H}_{ \star }(X_{m_{\star}})$   the covector in $T_{m_{\star}}^{\star}M$ acting as 
\be\label{eqn: for C-R-5}
\langle\mathrm{H}_{ \star }(X_{m_{\star}}),\,Z_{m_{\star}}\rangle\,:=\,\mathrm{H}_{ \star }(Z_{m_{\star}},\,X_{m_{\star}})\;\;\;\forall \,Z_{m_{\star}}\in\,T_{m_{\star}}M,
\ee
and by $\Gg^{M}\left(  Y^{0}_{m_{\star}}\right)$   the covector in $T_{m_{\star}}^{\star}M$ given by
\be\label{eqn: for C-R-6}
\langle \Gg^{M}\left( Y^{0}_{m_{\star}}\right),\,Z_{m_{\star}}\rangle\,:=\, \Gg^{M}\left(  Y^{0}_{m_{\star}},Z_{m_{\star}}\right)\;\;\;\forall \,Z_{m_{\star}}\in\,T_{m_{\star}}M\,.
\ee
Then, comparing equation \eqref{eqn: for C-R-2} with equation \eqref{eqn: for C-R-3}, equation \eqref{eqn: for C-R-5} and \eqref{eqn: for C-R-6} allows us to conclude that 
\be\label{eqn: for C-R-7}
\mathrm{H}_{ \star }(X_{m_{\star}})\,=\,\Gg^{M}\left( \lambda\,Y^{0}_{m_{\star}}\right),
\ee
which, assuming $\Gg^{M}$ to be invertible, is equivalent to 
\be\label{eqn: for C-R-8}
(\Gg^{M})^{-1}\left(\mathrm{H}_{ \star }(X_{m_{\star}}),\alpha_{m_{\star}}\right)\,=\,\langle\alpha_{m_{\star}},\,\lambda\,Y^{0}_{m_{\star}}\rangle
\ee
for all covectors $\alpha_{m_{\star}}\in T_{m_{\star}}^{\star}M$.
In particular, setting $\alpha_{m_{\star}}=\mathrm{H}_{ \star }(X_{m_{\star}})$ we get
\be\label{eqn: for C-R-9}
(\Gg^{M})^{-1}\left(\mathrm{H}_{ \star }(X_{m_{\star}}),\mathrm{H}_{ \star }(X_{m_{\star}})\right)\,=\,\langle\mathrm{H}_{ \star }(X_{m_{\star}}),\,\lambda\,Y^{0}_{m_{\star}}\rangle\,=\,\lambda\,H( Y^{0}_{m_{\star}})
\ee
because of equation \eqref{eqn: for C-R-5} and \eqref{eqn: for C-R-2}.
Now, equation \eqref{eqn: for C-R-2} together with equation \eqref{eqn: for C-R-4} and equation \eqref{eqn: for C-R-9} imply that
\be\label{eqn: almost cramer rao 3}
\begin{split}
\left(\mathrm{H}_{ \star }(Y_{m_{\star}}^{0},\,X_{m_{\star}})\right)^{2}&\,=\,\left(H(Y_{m_{\star}}^{0})\right)^{2}\,=\, \lambda\,\,H(Y_{m_{\star}}^{0})\,=\, \\
&\,=\,\left(\Gg^{M}\right)^{-1}\,\left( \mathrm{H}_{m_{\star}}(X_{m_{\star}}),\,\mathrm{H}_{m_{\star}}(X_{m_{\star}})\right)\,.
\end{split} 
\ee
Eventually, recalling that $Y_{m_{\star}}^{0}$ lies on the Fisher-Rao unit sphere,  equation \eqref{eqn: almost cramer rao 2} and \eqref{eqn: almost cramer rao 3} lead us to the generalized Cramer-Rao bound
\be\label{eqn: C-R bound}
\mathcal{C}(X_{m_{\star}},X_{m_{\star}})\,\geq\,\left(\Gg^{M}\right)^{-1}\,\left(\mathrm{H}_{ \star }(X_{m_{\star}}),\,\mathrm{H}_{ \star }(X_{m_{\star}})\right)\,.
\ee
If the Hessian form of $L_{\star}$ at $m_{\star}$ is invertible, we define the covariance bivector $\mathrm{Cov}$ as
\be
\label{eqn: Covdef}
\mathrm{Cov}(\xi_{m_{\star}},\eta_{m_{\star}})\,:=\,\mathcal{C}\left(\mathrm{H}_{ \star }^{-1}(\xi_{m_{\star}}),\,\mathrm{H}_{ \star }^{-1}(\eta_{m_{\star}})\right),
\ee
where $\xi_{m_{\star}},\eta_{m_{\star}}\in T_{m_{\star}}^{\star}M$.  We may then rewrite the generalized Cramer-Rao bound in terms of covectors. 
We proved the following:

\begin{proposition}\label{prop: Cramer-Rao bound}
Let $(M,\mathit{j},\Delta_{n}^{+})$ be a  parametric statistical model for which $\Gg^{M}$ is invertible.
Let $C$ be a cost function and let $\mathcal{E}$ be a stationary estimator for $C$ at $m_{\star}$.
If the Hessian form of $L_{\star}$ at $m_{\star}$ is invertible, then we have the  generalized Cramer-Rao bound 
\be\label{eqn: Cramer-Rao bound}
\mathrm{Cov}(\xi_{m_{\star}},\xi_{m_{\star}})\geq\,\left(\Gg^{M}\right)^{-1}\,\left(\xi_{m_{\star}},\,\xi_{m_{\star}}\right)\,
\ee
for all $\xi_{m_{\star}}\in T^{\star}_{m{\star}}M$.
\end{proposition}

Let us stress that, because of the  regularity properties satisfied by the model $(M,\mathit{j},\Delta_{n}^{+})$ (see definition \ref{def: parametric statistical model}, the paragraph right after it, and remark \ref{ex: classical statistical models}) and because of the assumed invertibility  of $\Gg^{M}$, the formulation of the Cramer-Rao bound given in proposition \ref{prop: Cramer-Rao bound} refers to the case in which the support of the considered probability distributions does not depend on the element $m$ in the parameter manifold $M$.
It is worth noting that in the literature, when the support of the considered probability density functions may depend on the parameter, there is still a version of the Cramer-Rao bound, the so-called Cramer-Rao-Leibniz bound, see for instance \cite{L-BS-W-P-D-2020}.

A stationary estimator $\mathcal{E}$ which saturates the Cramer-Rao bound for every $v_{m}$ is called \grit{efficient}.
The Cramer-Rao bound is related to the cost function $C$ and to the estimator $\mathcal{E}$, however, it is expressed in terms of the (inverse of the) Fisher-Rao metric tensor on $M$ which is a geometrical object on $M$ which is completely independent of the cost function and the estimator.
Note, however, that the expression \eqref{eqn: Covdef} is invariant under rescaling the cost function $C$, because  the expression $\mathcal{C}$ by \eqref{eqn: C-R 02} contains such a scaling factor quadratically, and this is cancelled  because the inverse of the Hessian enters quadratically into \eqref{eqn: Covdef}.

\begin{remark}[The Cramer-Rao bound for Euclidean cost functions]\label{ex: C-R bound for statistical models of open submanifolds of Euclidean space}
 
The ``standard form'' of the Cramer-Rao inequality used in classical information geometry is obtained when we  $M$ and the cost function $C$ are  as in remark \ref{subsec: stationary estimators for Euclidean cost function}.
In this case, a direct computation shows that, in local coordinates around $m_{\star}$,  the components Hessian form of $L_{\star}$ at every stationary point  are given by
\be
\left(\mathrm{H}_{\star}\right)_{jk}\,=\,\left(\delta_{rs}\,\frac{\partial m^{r}}{\partial \theta^{j}}\frac{\partial m^{s}}{\partial \theta^{k}}\right)(m_{\star}).
\ee

Assuming that $M$ is open in the ambient manifold $\mathbb{R}^{N}$, and taking $\{\theta^{1},...,\theta^{N}\}$ to be the Cartesian coordinates associated with the canonical projections of $\mathbb{R}^{N}$ on $\mathbb{R}$ we immediately see that
\be
\left(\mathrm{H}_{\star}\right)_{jk}\,=\, \delta_{jk} .
\ee
Therefore, writing
\be
\left(\mathrm{Cov}(m_{\star})\right)^{jk}\,\equiv\,\mathrm{Cov}(\mathrm{d}\theta^{j}(m_{\star}), \mathrm{d}\theta^{k}(m_{\star})),
\ee
a direct computation shows that the covariance matrix $\left(\mathrm{Cov}(m_{\star})\right)^{jk}$ at the point $m_{\star}$ for which $\mathcal{E}$ is a stationary estimator reads
\be\label{eqn: covariance for open submanifold of vector space  with Euclidean cost function}
\begin{split}
\left(\mathrm{Cov}(m_{\star})\right)^{jk}&\,=\,\mathrm{E}_{p_{\star}}\left[\left(\mathcal{E}^{j} - \mathrm{E}_{p_{\star}}\left[\mathcal{E}^{j}\right]\right)\,\left(\mathcal{E}^{k} - \mathrm{E}_{p_{\star}}\left[\mathcal{E}^{k}\right]\right)\right]
\end{split},
\ee
which is essentially the form usually found in standard textbooks on estimation theory in statistics.
The ``standard form'' of the Cramer-Rao bound follows immediately. 
\end{remark}

\section{The Helstrom bound}\label{subsec: helstrom bound for manifold-valued estimators}

The Cramer-Rao bound found in section \ref{sec: cramer-rao bound for manifold-valued estimators} applies to \grit{parametric statistical models}. 
As such, it depends only on the Fisher-Rao metric on $M$ which, in turn, depends on the properties of the Abelian algebra underlying the parametric statistical model.
Accordingly, if $(M,\mathit{j}^{c},\Delta_{n}^{+})$ is the parametric statistical model associated with a parametric model of states $(M,\mathit{j},\mathcal{O})$ on the possibly noncommutative $C^{\star}$-algebra $\appa$, the Cramer-Rao bound for $(M,\mathit{j}^{c},\Delta_{n}^{+})$ ``does not feel'' the possible noncommutativity of the algebra $\appa$.
However, it is possible to formulate a bound which ``feels'' the possible non-commutativity of $\appa$, and this bound is related with the metric tensor $\Gg^{M}$ and its relation with $\Gg^{Mc}$.
This bound is essentially the $C^{\star}$-algebraic formulation of the Helstrom bound used in quantum information theory, and the content of the following proposition will be the key point to formulate the Helstrom bound in the $C^{\star}$-algebraic framework.

\begin{proposition}\label{prop: Jordan metric majorizes classical jordan metric}
Let $(M,\mathit{j},\mathcal{O})$ be a parametric model of states on the finite-dimensional $C^{\star}$-algebra $\appa$, and let  $\Gg^{M}$ be the symmetric covariant tensor on $M$ defined by equation \eqref{eqn: pullback metric}.
Let  $(M,\mathit{j}^{c},\Delta_{n}^{+})$ be a parametric statistical   model associated with  $(M,\mathit{j},\mathcal{O})$,  and let $\Gg^{Mc}$ be the symmetric covariant tensor on $M$ defined by equation \eqref{eqn: classical pullback metric}.
Then, we have
\be
\Gg^{M}_{m}(v_{m},v_{m})\,\geq\,\Gg^{Mc}_{m}(v_{m},v_{m})
\ee
for every $m\in M$ and every $v_{m}\in T_{m}M$.
\end{proposition}
\begin{proof}
According to the definition of the SLD given in equation \eqref{eqn: SLD}, given an arbitrary tangent vector $v_{m}\in T_{m}M$, there is a gradient vector field $\mathbb{Y}_{\mathbf{a}}$ on $\mathcal{O}$ such that
\be
T_{m}\mathit{j}(v_{m})\,=\,\mathbb{Y}_{\mathbf{a}}(\rho_{m}).
\ee
Consequently, we have  (recalling \eqref{eqn: Jordan metric tensor}) 
\be
\Gg^{M}_{m}(v_{m},v_{m})\,=\,\Gg_{\rho_{m}}(\mathbb{Y}_{\mathbf{a}}(\rho_{m}),\mathbb{Y}_{\mathbf{a}}(\rho_{m}))\,=\,\rho_{m}(\mathbf{a}^{2}) - \left(\rho_{m}(\mathbf{a})\right)^{2}\,.
\ee

On the other hand, by definition, we have
\be
\Gg^{Mc}=(\mathit{j}^{c})^{*}\Gg_{FR}\,=\,(\mathfrak{m}^{\star}\,\circ\,\mathit{j})^{*}\Gg_{FR}\,=\,\mathit{j}^{*}\left((\mathfrak{m}^{\star})^{*}\Gg_{FR}\right)\,,
\ee
which means 
\be
\Gg^{Mc}_{m}(v_{m},v_{m})\,=\,\left((\mathfrak{m}^{\star})^{*}\Gg_{FR}\right)_{\rho_{m}}(\mathbb{Y}_{\mathbf{a}}(\rho_{m}),\mathbb{Y}_{\mathbf{a}}(\rho_{m}))\,,
\ee
and thus we have to prove that
\be\label{eqn: jordan metric maj FR 1}
\left((\mathfrak{m}^{\star})^{*}\Gg_{FR}\right)_{\rho_{m}}(\mathbb{Y}_{\mathbf{a}}(\rho_{m}),\mathbb{Y}_{\mathbf{a}}(\rho_{m}))\,\leq\,\rho_{m}(\mathbf{a}^{2})  - \left(\rho_{m}(\mathbf{a})\right)^{2}
\ee
to prove the proposition.

We note that, fixed  any $\rho\in\mathcal{O}$ and given an arbitrary non-zero gradient tangent vector  $\mathbb{Y}_{\mathbf{a}}(\rho)$, there is an element $\mathbf{a}_{c}\in\cappa_{n} \equiv C(\mathcal{X}_{n})$ and a gradient tangent vector $\mathbb{Y}_{a_{c}}(\mathfrak{m}^{\star}(\rho))$ at $\mathfrak{m}^{\star}(\rho)\in\mathcal{O}^{c}$ such that
\be\label{eqn: gradient tangent vector 2}
T_{\rho}\mathfrak{m}^{\star}(\mathbb{Y}_{\mathbf{a}}(\rho))\,=\,\mathbb{Y}_{a_{c}}(\mathfrak{m}^{\star}(\rho)),
\ee
and a direct computation shows that $\mathbf{a}_{c}$ is characterized by the property
\be\label{eqn: gradient tangent vector 1}
\rho(\{\mathbf{a},\mathfrak{m}(\mathbf{b}_{c})\}) - \rho(\mathbf{a})\,\rho(\mathfrak{m}(\mathbf{b}_{c}))\,=\,\rho(\mathfrak{m}(\mathbf{a}_{c}\,\mathbf{b}_{c})) - \rho(\mathfrak{m}(\mathbf{a}_{c}))\,\rho(\mathfrak{m}(\mathbf{b}_{c}))
\ee
for all $\mathbf{b}_{c}\in\cappa_{n}$.
Therefore, we have
\be
\begin{split}
\left((\mathfrak{m}^{\star})^{*}\Gg_{FR}\right)_{\rho}(\mathbb{Y}_{\mathbf{a}}(\rho),\mathbb{Y}_{\mathbf{a}}(\rho))&\,=\, (\Gg_{FR})_{\mathfrak{m}^{\star}(\rho)}\left(T_{\rho}\mathfrak{m}^{\star}(\mathbb{Y}_{\mathbf{a}}(\rho)),\,T_{\rho}\mathfrak{m}^{\star}(\mathbb{Y}_{\mathbf{a}}(\rho))\right)\,=\, \\
&\,=\, (\Gg_{FR})_{\mathfrak{m}^{\star}(\rho)}\left(\mathbb{Y}_{a_{c}}(\mathfrak{m}^{\star}(\rho)),\,\mathbb{Y}_{a_{c}}(\mathfrak{m}^{\star}(\rho))\right)\,=\, \\
&\,=\,\rho(\mathfrak{m}(\mathbf{a}_{c}^{2})) - \left(\rho(\mathfrak{m}(\mathbf{a}_{c}))\right)^{2}\,.
\end{split}
\ee
Recalling equation \eqref{eqn: jordan metric maj FR 1}, we see that if the inequality
\be
\rho(\mathfrak{m}(\mathbf{a}_{c}^{2})) - \left(\rho(\mathfrak{m}(\mathbf{a}_{c}))\right)^{2}\,\leq\,\rho(\mathbf{a}^{2})  - \left(\rho(\mathbf{a})\right)^{2}
\ee
holds for all $\rho,\mathbf{a},\mathbf{a}_{c}$ and $\mathfrak{m}$ satisfying equation \eqref{eqn: gradient tangent vector 2}, then the proposition is proved.

Next, by means of equation \eqref{eqn: gradient tangent vector 1}, we write
\be
\rho(\mathfrak{m}(\mathbf{a}_{c}^{2})) - \left(\rho(\mathfrak{m}(\mathbf{a}_{c}))\right)^{2}\,=\,\rho(\{\mathbf{a},\mathfrak{m}(\mathbf{a}_{c})\}) - \rho(\mathbf{a})\,\rho(\mathfrak{m}(\mathbf{a}_{c}))\,,
\ee
and since $\rho(\{\cdot,\cdot\}) - \rho(\cdot)\,\rho(\cdot)$ is an inner product on the space of self-adjoint elements of $\appa$, we may apply the Cauchy-Schwarz inequality to obtain
\be
\left(\rho(\mathfrak{m}(\mathbf{a}_{c}^{2})) - \left(\rho(\mathfrak{m}(\mathbf{a}_{c}))\right)^{2} \right)^{2} \,\leq\,\left(\rho(\mathbf{a}^{2})-(\rho(\mathbf{a}))^{2}\right)\,\, \left(\rho(\mathfrak{m}(\mathbf{a}_{c})\,\mathfrak{m}(\mathbf{a}_{c}))  - \left(\rho(\mathfrak{m}(\mathbf{a}_{c}))\right)^{2} \right)\,.
\ee
Now, $\mathfrak{m}$ is a positive  unital map, and thus it  satisfies Kadison's inequality 
\be
\mathfrak{m}(\mathbf{a}_{c}^{2})\,\geq \mathfrak{m}(\mathbf{a}_{c})\,\mathfrak{m}(\mathbf{a}_{c}), 
\ee
from which it follows that
\be
\rho(\mathfrak{m}(\mathbf{a}_{c}^{2}))\,\geq\,\rho(\mathfrak{m}(\mathbf{a}_{c})\,\mathfrak{m}(\mathbf{a}_{c}))\,.
\ee
Consequently, assuming that $\rho(\mathfrak{m}(\mathbf{a}_{c}^{2})) - \left(\rho(\mathfrak{m}(\mathbf{a}_{c}))\right)^{2} \neq 0$, we have
\be
\frac{\rho(\mathfrak{m}(\mathbf{a}_{c})\,\mathfrak{m}(\mathbf{a}_{c}))  - \left(\rho(\mathfrak{m}(\mathbf{a}_{c}))\right)^{2}}{\rho(\mathfrak{m}(\mathbf{a}_{c}^{2})) - \left(\rho(\mathfrak{m}(\mathbf{a}_{c}))\right)^{2}}\,\leq\,1
\ee
and thus
\be
\begin{split}
\rho(\mathfrak{m}(\mathbf{a}_{c}^{2})) - \left(\rho(\mathfrak{m}(\mathbf{a}_{c}))\right)^{2}& \,\leq\, \rho(\mathbf{a}^{2})-(\rho(\mathbf{a}))^{2}\,,
\end{split}
\ee
and the proposition is proved.

\end{proof}

From the proof of proposition  \ref{prop: Jordan metric majorizes classical jordan metric}, we easily obtain the following corollary.

\begin{corollary}\label{corollary: attainability of Helstrom bound}
Let $(M,\mathit{j},\mathcal{O})$ be a parametric model of states on the $C^{\star}$-algebra $\appa$.
Suppose there is a unital, Abelian $C^{\star}$-subalgebra $\cappa\subseteq\appa$ such that, for all $v_{m}\in T_{m}M$, the SLD $\mathbb{Y}_{\mathbf{a}}(\rho_{m})$ of $v_{m}$ at $\rho_{m}=\mathit{j}(m)$ given by 
\be
T_{m}\mathit{j}(v_{m})\,=\,\mathbb{Y}_{\mathbf{a}}(\rho_{m})
\ee
is such that $\mathbf{a}\in\cappa$.
Suppose also that the measurement procedure $\mathfrak{m}:=\mathit{i}_{\cappa}$ given by the natural inclusion of $\cappa$ in $\appa$ gives rise to a parametric statistical model $(M,\mathit{j}^{c},\Delta_{n}^{+})$ associated with $(M,\mathit{j},\mathcal{O})$.
Then, it holds
\be
\Gg^{M}_{m}(v_{m},v_{m})\,=\,\Gg^{Mc}_{m}(v_{m},v_{m}).
\ee

\end{corollary}

Now, let $(M,\mathit{j},\mathcal{O})$ be a parametric model of states on the $C^{\star}$-algebra $\appa$, and let $(M,\mathit{j}^{c},\Delta_{n}^{+})$ be a parametric statistical model associated with  $(M,\mathit{j},\mathcal{O})$.
Assume $(M,\mathit{j},\mathcal{O})$ and $(M,\mathit{j}^{c},\Delta_{n}^{+})$ to be such that $\Gg^{M}$ and $\Gg^{Mc}$  are invertible.
Let $C$ be a cost function and $\mathcal{E}$ an estimator as in section \ref{sec: estimation theory}.
Assume $\mathcal{E}$ is a stationary estimator at $m_{\star}$, and let $C_{\mathcal{E}_{j}}\colon M\ra\mathbb{R}$ be the smooth function given by $C_{\mathcal{E}_{j}}(m):=C(m,\mathcal{E}_{j})$, where $\mathcal{E}_{j}\equiv \mathcal{E}(\mathbf{x}_{j})$ with $\mathbf{x}_{j}\in\mathcal{X}_{n}$.

According to the results of section \ref{sec: cramer-rao bound for manifold-valued estimators} (see equations \eqref{eqn: C-R 01}, \eqref{eqn: C-R 02}, and \eqref{eqn: C-R bound}), given $v_{m_{\star}},w_{m_{\star}}\in T_{m_{\star}}M$, the bilinear form 
\be
\mathcal{C}(v_{m_{\star}},w_{m_{\star}})\,:=\,\sum_{j=1}^{n}\,v_{m_{\star}}(C_{\mathcal{E}_{j}})\,w_{m_{\star}}(C_{\mathcal{E}_{j}})\,p^{j}(m_{\star}),
\ee
where $v_{m_{\star}}(C_{\mathcal{E}_{j}})$ is the  derivative of $C_{\mathcal{E}_{j}}$  in the direction of $v_{m_{\star}}$ evaluated at $m_{\star}\in M$ (and similarly for $w_{m_{\star}}(C_{\mathcal{E}_{j}})$), satisfies the Cramer-Rao bound given by 
\be
\mathcal{C}(v_{m_{\star}},v_{m_{\star}})\,\geq\, \left(\Gg^{Mc}_{m_{\star}}\right)^{-1}\,\left(\mathrm{H}_{m_{\star}}(v_{m_{\star}}),\,\mathrm{H}_{m_{\star}}(v_{m_{\star}})\right),
\ee
where $\Gg^{Mc}$ is the Fisher-Rao metric on $M$ seen as a parametric statistical model in $\Delta_{n}^{+}$, and $\mathrm{H}_{m_{\star}}$ is the Hessian form  of the function $L_{m_{\star}}\colon M \ra \mathbb{R}$ given by $L_{m_{\star}}(m_{1}):=L(m_{1},m_{\star})$ at the point $m_{1}=m_{\star}$ (see equation \eqref{eqn: parametric loss functional}).

Then, proposition \ref{prop: Jordan metric majorizes classical jordan metric} states that 
\be
 \Gg^{M}_{m}(w_{m},w_{m})\,\geq\,\Gg^{Mc}_{m}(w_{m},w_{m})
\ee
for every $w_{m}\in T_{m}M$.
Consequently, we also obtain that 
\be
\left(\Gg_{m}^{M }\right)^{-1}(\alpha_{m},\alpha_{m})\,\leq\, \left(\Gg_{m}^{Mc}\right)^{-1}(\alpha_{m},\alpha_{m})
\ee
for every $\alpha_{m}\in T_{m}^{\star}M$ (see \cite[Ex. 1.2.12]{Bhatia-2007}), and the Cramer-Rao bound in equation \eqref{eqn: Cramer-Rao bound} allows us to state that
\be\label{eqn: Helstrom bound}
\mathcal{C}(v_{m_{\star}},v_{m_{\star}})\,\geq\, \left(\Gg_{m_{\star}}^{Mc}\right)^{-1} \,\left(\mathrm{H}_{m_{\star}}(v_{m_{\star}}),\,\mathrm{H}_{m_{\star}}(v_{m_{\star}})\right)\,\geq\, \left(\Gg_{m_{\star}}^{M}\right)^{-1} \,\left(\mathrm{H}_{m_{\star}}(v_{m_{\star}}),\,\mathrm{H}_{m_{\star}}(v_{m_{\star}})\right).
\ee
We proved the following:

\begin{proposition}
Let $(M,\mathit{j}^{c},\Delta_{n}^{+})$ be the parametric statistical model associated with a parametric model of states $(M,\mathit{j},\mathcal{O})$.
Assume that both $\Gg^{M}$ and $\Gg^{Mc}$ are invertible. 
Let $C$ be a cost function and let $\mathcal{E}$ be a stationary estimator for $C$ at $m$.
If the Hessian form of $L_{m_{\star}}$ at $m_{\star}$ is invertible, then we have the  generalized Helstrom bound  
\be\label{eqn: Helstrom bound 2}
\mathrm{Cov}(\xi_{m_{\star}},\xi_{m_{\star}})\,\geq\,\left(\Gg^{Mc}_{m_{\star}}\right)^{-1}\,\left(\xi_{m_{\star}},\,\xi_{m_{\star}}\right)\geq\,\left(\Gg^{M}_{m_{\star}}\right)^{-1}\,\left(\xi_{m_{\star}},\,\xi_{m_{\star}}\right)
\ee
for all $\xi_{m_{\star}}\in T^{\star}_{m{\star}}M$.
\end{proposition}

This is the \grit{Helstrom bound} for parametric models of states on a $C^{\star}$-algebra.
Indeed,  when $\appa$ is the algebra $\bh$ of bounded operators on the Hilbert space $\hh$ of a finite-level quantum system, $\mathcal{O}$ is the orbit of faithful density operators on $\hh$, $M$ is an open subset of some $\mathbb{R}^{k}$ with $k\in\mathbb{N}$.
Then, in accordance with remark \ref{subsec: stationary estimators for Euclidean cost function}, the cost function $C$ may be taken to be the Euclidean distance on $\mathbb{R}^{k}\times\mathbb{R}^{k}$ pulled back on $M\times M$, and a direct computation shows that  equation \eqref{eqn: Helstrom bound}  reduces to the so-called \grit{Helstrom  bound}  used in quantum estimation theory or quantum metrology  \cite{Helstrom-1967,Helstrom-1968,Helstrom-1969,Paris-2009}.

\begin{remark}[Helstrom bound for multiple-round models]
If we consider multiple rounds as in the end of section \ref{sec: statistical models}, that is, we set $\mathcal{X}=\mathcal{Y}^{N}$,  then proposition \ref{prop: Jordan metric for N rounds is N times the Jordan metric} implies that the Helstrom bound can be written as
\be\label{eqn: Helstrom bound on multiple rounds}
\begin{split}
\mathcal{C}(v_{m},v_{m})&\geq \left(\Gg_{m}^{McN}\right)^{-1} \left(\mathrm{H}_{m}(v_{m}),\mathrm{H}_{m}(v_{m})\right)\\
&\geq\left(\Gg_{m}^{MN}\right)^{-1} \left(\mathrm{H}_{m}(v_{m}),\mathrm{H}_{m}(v_{m})\right)\\
&\geq\frac{1}{N}\,\left(\Gg_{m}^{M}\right)^{-1}\,\left(\mathrm{H}_{m}(v_{m}),\mathrm{H}_{m}(v_{m})\right),
\end{split}
\ee
and this equation allows for the asymptotic analysis of the bound.
\end{remark}

The Helstrom bound is a universal bound for all the possible parametric statistical models associated with a given parametric model of states on a given $C^{\star}$-algebra.
This makes it quite a remarkable bound.

It is clear that, independently of the cost function and of the estimator we may choose, the Helstrom bound may be saturated if and only if 
\be
\Gg^{M}_{m}(v_{m},v_{m})\,=\,\Gg^{Mc}_{m}(v_{m},v_{m}).
\ee
Then, corollary \ref{corollary: attainability of Helstrom bound} shows that this is in principle always true for one-dimensional models because we can always take the unital, Abelian $C^{\star}$-subalgebra  generated by the self-adjoint element $\mathbf{a}$ associated with the SLD of a given $v_{m}$ at $\rho_{m}$, and we are in the hypothesis of the corollary.
However, it is also clear that for higher-dimensional models like the one in example \ref{ex: a mixed state qubit model}, this strategy may not be available.

\section{Conclusion}

We presented a preliminary account of the formulation of estimation theory in the context of parametric models of states on finite-dimensional $C^{\star}$-algebras.
The aim is to set the stage for the development of a mathematical formulation of estimation theory that is able to deal with the classical and quantum case ``at the same time'' by simply switching between commutative and noncommutative algebras.

After reviewing the differential geometric properties of the space of states $\stsp$ of an arbitrary finite-dimensional $C^{\star}$-algebra $\appa$, we introduced the notion of parametric model of states on $\appa$.
Then, following what is done in quantum information theory using POVMs,  we considered how the explicit choice of a positive linear map from $\appa$ to a suitable commutative $C^{\star}$-algebra $\cappa$ gives rise to the notion of parametric statistical model of states associated with the  starting parametric model of states on $\appa$.
This parametric statistical model may be viewed as a classical-like snapshot of the given parametric model of states on the possibly noncommutative algebra $\appa$, and the Cramer-Rao bound for manifold-valued estimators  is available for this model.

The fact that when $\appa$ is noncommutative there is more than one such classical-like snapshot means that  there is a Cramer-Rao bound for every classica-like snapshot of a given parametric model of states on  $\appa$.
This instance leads us to  reformulate the so-called Helstrom bound  to the case of a parametric model of states on a generic $C^{\star}$-algebra and not just the algebra of bounded linear operators on a Hilbert space as it is customarily done in quantum information theory. 
The Helstrom bound gives a lower bound for all the possible Cramer-Rao bounds associated with the classical-like snapshots of a given parametric model of states on $\appa$. 
The possibility of considering also multiple-round models is briefly discussed, and the Helstrom bound derived in this context will be the starting point for the asymptotic theory of estimation theory in the $C^{\star}$-algebraic framework we will deal with in future works.

As already remarked in the introduction, this work should be interpreted as a  preliminary step toward a more general understanding of  classical and quantum estimation theory.
Accordingly, there are different instances  that are left open  for further developments.
For instance, it is necessary to understand the general conditions for the attainability of the Helstrom bound for parametric models of states of dimension greater or equal than 2.
It is also necessary to understand how to formulate other relevant bounds like the RLD-bound and the Holevo bound used in quantum information theory in the $C^{\star}$-algebraic framework,  as well as to understand how to perform the transition  to the infinite dimensional case.
From another point of view, it would be interesting to understand a suitable $C^{\star}$-algebraic counterpart of the Amari-Cencov 3-tensor and the affine geometry it encodes in order generalize to the quantum case the understanding of the role of Frobenius manifolds recently investigated in the classical case \cite{C-M-2020,J-T-Z-2020}.
We plan to address these issues in future works.

\footnotesize

\addcontentsline{toc}{section}{References}

\end{document}